\documentclass[sn-mathphys]{sn-jnl}

\usepackage{makecell}

\newcommand{\tr}{{\mathrm{Tr}}}
\newcommand{\gf}{{\mathbb{F}}}

\jyear{2021}%

\theoremstyle{thmstyleone}%
\newtheorem{theorem}{Theorem}

\newtheorem{definition}{Definition}
\newtheorem{lemma}[theorem]{Lemma}

\theoremstyle{thmstyletwo}%

\theoremstyle{thmstylethree}%

\begin{document}

\title[On the $ c $-differential spectrum of power functions over finite fields]{On the $ c $-differential spectrum of power functions over finite fields}

\author[1]{\fnm{Kun} \sur{Zhang}}\email{kzhang@my.swjtu.edu.cn}

\author*[1]{\fnm{Haode} \sur{Yan}}\email{hdyan@swjtu.edu.cn}


\affil[1]{\orgdiv{School of Mathematics},  \orgaddress{\street{Southwest Jiaotong University}, \city{Chengdu}, \postcode{610031}, \country{China}}}


\abstract{Recently, a new concept called multiplicative differential was introduced by Ellingsen \textit{et al}. Inspired by this pioneering work, power functions with low $ c $-differential uniformity were constructed. Wang \textit{et al.} defined the $ c $-differential spectrum of a power function \cite{wang2021several}. In this paper, we present some properties of the $ c $-differential spectrum of a power function. Then we apply these properties to investigate the $ c $-differential spectra of some power functions. A new class of AP$c$N function is also obtained.}




\maketitle

\section{Introduction}\label{sec1}

Differential cryptanalysis \cite{1991Differential, 1993DifferentialBS} is the first statistical attack to decipher iterative block ciphers. Substitution boxes (S-boxes for short) play an crucial role in the field of symmetric block ciphers, which can be seen as cryptographic
functions over finite fields. Let $ {\mathbb{F}_{q}} $ be the finite fields with $ q $ elements. For a function $F(x)$ from $\gf_{q}$ to itself, the main tools to handle $F$ regarding the differential
attack are the Difference Distribution Table (DDT for short) and the differential uniformity $\Delta_F$, introduced by Nyberg \cite{NK1993} in 1994. For any $ a, b \in \gf_{2^n}$, the DDT entry at point $ (a,b) $, denoted by $ \delta_F(a,b) $, is defined as \[{\delta _F}\left( {a,b} \right) = \lvert\{x \in \gf_{q}: ~F(x+a)-F(x)=b\}\rvert, \]
where $ \lvert S \rvert $ denotes the cardinality of a set $S$. The differential uniformity of the function $ F(x) $, denoted by $ {\Delta _F} $, is defined as
\[{\Delta _F}=\max\limits_{a \in \gf_{q}^*}\max\limits_{b \in \gf_{q}}  \delta_F(a,b) ,\]
where $\gf_{q}^*=\gf_{q}\setminus\{0\}$. When $F$ is used as an S-box inside a cryptosystem, the smaller the value $\Delta_F$ is, the better $F$ to the resistance against differential attack. When $\Delta_{F}=1$, $F$ is called perfect nonlinear (PN for short) function. When $\Delta_F=2$, $F$ is called almost perfect nonlinear (APN for short) function. Note that PN functions over even characteristic finite fields do not exist. PN and APN functions play an important role in both theory and applications.
The recent progress can be found in \cite{CM1997,DNiho,DO1968,DWelch,HRS,HS,ZW2009,ZKW2009} and their references.
Power functions with low differential uniformity serve as good candidates for the design of S-boxes not only because of their strong resistance to differential attacks but also for the usually low implementation cost in hardware. When $F(x)$ is  a power function, i.e.,  $F(x)=x^d$ for an integer $d$,  one can easily see that $\delta_F(a,b)=\delta_F(1,{b/{a^d}})$  for all $a\in \gf_{q}^*$ and $b\in \gf_{q}$.
That is to say,  the differential properties of $F(x)$ is completely determined by the values of $\delta_F(1,b)$ as $b$ runs through $\gf_{q}$.
Therefore, Blondeau, Canteaut and Charpin first defined the difference spectrum of a power function in \cite{BCC}. Denote
\[\omega_i= \lvert \left\{b\in \gf_{q}: \delta_F(1, b)=i\right\} \rvert ,\,\,0\leq i\leq \Delta_F,\]
where $\Delta_{F}$ is the differential uniformity of $F$. The differential spectrum of $F$ is defined as the multi-set
\[ DS_F =\left\{\omega_i>0~:~0\leq i \leq \Delta_F \right\}.\]

The distribution of DDT of a power function can be deduced via its differential spectrum. Moreover, the differential spectrum is also an important notion for estimating its resistance against variants of differential cryptanalysis (\cite{BCC}, \cite{bracken2010highly}, \cite{charpin2014sparse}). However, it is difficult to completely determine the differential spectrum of a power function. For a known results on differential spectrum of power functions, the readers are referred to \cite{BCC, BCC2, BP, XYY, Li, Dobbertin2001, 2020XZL, CHNC, Yan, LRF,XY2017,YXLHXL}. 
The distribution of DDT of a power function can be deduced via its differential spectrum.  

In \cite{2002BN}, the authors used a new type of differential, namely multiplicative differential, that is quite useful from a practical perpective for ciphers that utilize modular multiplication as a primitive operation. It is an extension of differential cryptanalysis, and it cryptanalyzes some existing ciphers (like a variant of the well-known IDEA cipher). The authors argue that one should look at other types of differential for a cryptographic function $ F $, not only the usual $ (F(x+a),F(x)) $ but $ (F(x+a),cF(x)) $. Moreover, they first introduced the $c$-Differential  Distribution Table ($c$DDT for short) . For a function $F$ from $\gf_{{q}}$ to itself and $c\in\gf_{{q}}$, the entry at point $(a,b)$ of the $c$DDT, denoted by $_c\delta_F(a,b)$, is defined as

\[_c\delta_F(a,b)=\#\{x\in\gf_q: F(x+a)-cF(x)=b\}.\]
The corresponding $c$-differential uniformity is defined as follows.
\begin{definition}(\cite{2020CEP})
	Let $ \gf_{q} $ denote the finite field with $ q $ elements, where $ q $ is a prime power. For a function $ F:{\gf_{q}} \to {\gf_{q}} $, and $a, b,  c\in \gf_{q} $, we call 
	\[_c{\Delta _F} = \max \left\{ {_c{\delta _F}\left(a, b \right):a,b \in {\gf_{{q}}},\mathrm{and}~a \ne 0~\mathrm{if}~c = 1} \right\}\]
	the $ c $-differential uniformity of $ F $.
\end{definition}
If $ _c{\Delta _F}=\delta $, then we say that $ F $ is differentially $ (c,\delta) $-uniform. Similarly, the smaller the value $_c\Delta_F$ is, the better $F$ 
to the resistance against multiplicative differential attack.  If the $ c $-differential uniformity of $ F $ equals $ 1 $, then $ F $ is called a perfect $ c $-nonlinear (P$ c $N) function. P$c$N function over odd characteristic finite fields are also called $ c $-planar functions. If the $ c $-differential uniformity of $ F $ is $ 2 $, then $ F $ is called an almost perfect $ c $-nonlinear (AP$ c $N) function. It is easy to conclude, for $ c=1 $ and $ a\ne 0 $, the $ c $-differential uniformity becomes the usual differential uniformity, and the P$ c $N and AP$ c $N functions become PN and APN functions  respectively. It is know that APN functions over finite fields of even characteristic have lowest differential uniformity. However, for the $ c $-differential uniformity, there exist P$c$N functions over even characteristic finite fields. There are a few functions with low $c$-differential uniformity reported. The readers can refer to \cite{wang2021several,2021SomeZha,tu2021class,yan2021on1,2021SYZ,2020CEP,2020MD,hasan2021c,WLZ,BC}.

Similarly, when $F(x)$ is a power function, one easily sees that $_c{\delta_F(a,b)}$=$_c{\delta_F(1,{b/{a^d}})}$  for all $a\in \gf_{{q}}^*$ and $b\in \gf_{q}$.
That is to say,  the $ c $-differential spectrum of $F(x)$ is completely determined by the values of $_c\delta_F(1,b)$ as $b$ runs through $\gf_{q}$.
Therefore,  the $ c $-differential spectrum of a power function can be defined as follows.

\begin{definition}(\cite{wang2021several}) Let $ F\left( x \right) = {x^d} $ be a power function over $\gf_{q} $ with $c$-differential uniformity $_c\Delta_F$. Denote by $ {}_c{\omega _i} = \# \{ {b \in \gf_{q} : {{}_c{\delta _F}( 1,b) = i} .} \} $ for each $ 0 \le i \le {}_c{\Delta _F} $. The $ c $-differential spectrum of $ F $ is defined to be the multi-set \[\mathbb{S} = \{ {{}_c{\omega _i}:{0 \le i \le {}_c{\Delta _F}} ~\mathrm{and}{}~_c{\omega _i} > 0} \}.\]
\end{definition}
When $c=1$, the $c$-differential spectrum becomes the usual differential spectrum. As far as we know, only two classes of power functions have known $ c $-differential spectra. In \cite{wang2021several}, the authors computed the $c$-differential spectrum of Gold function. The $c$-differential spectrum of $x^d$ over $\gf_{2^{4m}}$ was determined in \cite{tu2021class}, where $d=2^{3m}+2^{2m}+2^{m}-1$. {{We summarized the known results in Table \ref{table1} as well as the results obtained in this paper.} }The rest of this paper is organized as follows. Section II introduces some notation and useful lemmas. In Section III, we present the properties of $ c $-differential spectrum. In Section IV, we calculate the $ c $-differential spectrum of several power functions with low $ c $-differential uniformity. In Section V, we construct an infinite family of AP$ c $N power mappings over $ \gf_{{5^n}} $, and determine its $ c $-differential spectrum. Section VI concludes this paper.

\begin{table}[h]\label{table1}
	\begin{center}
		\begin{minipage}{\textwidth}
			\caption{Power functions $F(x)=x^d$ over $\gf_{p^n}$ with known $ c $-differential spectrum}\label{table-1}
			\begin{tabular}{@{}lllll@{}}
				\toprule
	$p$&$d$ & Condition & $_c\Delta_{F}$ & Ref.\\
				\midrule
				2& $2^{3m}+2^{2m}+2^{m}-1$&\makecell[l]{$n=4m$, $0,1\ne c\in \gf_{2^{n}}$,\\$c^{1+2^{2m}}=1$} 
				& $2$ &\cite{tu2021class}\\
				odd& $p^{k}+1$& \makecell[l]{$1\ne c\in \gf_{p^{\gcd(n,k)}}$ and $ \frac{n}{\gcd(n,k) }$ is odd, \\ or $ c\notin \gf_{p^{\gcd(n,k)}} $, $ n $ is even and $ k=\frac{n}{2} $} & $2$ &\cite{wang2021several}\\
				2& $2^{n}-2$& $c\ne 0$, $\tr_{n}(c)=\tr_{n}(c^{-1})=1$ &$2$ &Thm \ref{Inversep=2APcN}\\
				2&$2^n-2$ &$c\ne 0$, $\tr_{n}(c)=0$ or $\tr_{n}(c^{-1})=0$ & $3$  &Thm \ref{Inversep=2APcN}\\
				odd&$p^{n}-2$ &\makecell[l]{$c=4,4^{-1}$, or\\$\chi(c^{2}-4c)=-1$ and $\chi(1-4c)=-1$}  & 2 &Thm \ref{thm-oddAPcN1}\\
				odd&$p^{n}-2$ &\makecell[l]{$c\ne 0,4,4^{-1}$,\\$\chi(c^{2}-4c)=1 $ or $\chi(1-4c)=1$}  & 3  &Thm \ref{thm-oddAPcN1}\\
				3&$\frac{3^{n}+3}{2}$ &$c=-1$, $n$ even &$2$ &Thm \ref{p=3-niseven}\\
				3&$3^n-3$ &$c=-1$, $n = 0({\bmod 4})$ & 6  &Thm \ref{p=3d=-2}\\
				3&$3^n-3$ &$c=-1$, $n\ne0({\bmod 4})$ & 4  &Thm \ref{p=3d=-2}\\
				odd&$\frac{p^{k}+1}{2}$ & $c=-1$, $ \gcd(n,k)=1$, $k$ odd & $ \frac{p+1}{2} $ & Thm \ref{poddp+1/2-1}, \ref{poddp+1/2-3} \\
				5&$\frac{p^{n}-3}{2}$ & $c=-1$ & 2 & Thm \ref{p=5spectrum} \\
				\botrule
			\end{tabular}
			\footnotetext{ $ \tr_{n}\left(  \cdot  \right) $ denotes the absolute trace mapping from $ \gf_{{2^n}} $ to $ \gf_2 $.}
			\footnotetext{$ \chi\left(  \cdot  \right) $ denotes the quadratic multiplicative character on $ \gf_{{p^n}}^* $.}
		\end{minipage}
	\end{center}
\end{table}

\section{Preliminaries}
In this section, we first fix some notation and list some facts which will be used in this paper unless otherwise stated.

\begin{itemize}
	\item Let $\gf_q$ denote the finite field with $q$ elements.
	\item $\gf_{{q}}^*=\gf_{{q}}\setminus\{0\}$, $ \gf_{{q}}^\#  = {{\rm{\gf}}_{{q}}}\backslash \{ 0, - 1\}  $.
	\item $ \Delta_{c}(x)=(x+1)^{d}-cx^{d} $, where $ c\in \gf_{{q}} $.
	\item $ {\delta _c}\left( b \right) = \# \left\{ {x \in {\gf_{{q}}}:{\Delta _c}\left( x \right) = b} \right\} $.
	\item $ \chi $ denotes the quadratic multiplicative character on $ \gf_{{q}}$, i.e., 
	\[\chi(x) = \left\{ 
	\begin{aligned}
		1,~~&\mathrm{if}~{x}~\mathrm{a~square~},\\
		0,~~&\mathrm{if}~x=0,\\
		-1,~~&\mathrm{if}~{x}~\mathrm{a~nonsquare~}.
	\end{aligned} \right.\]
	It is well-known that $\sum_{x\in\gf_q}\chi(x)=0$.
	\item $ {S_{i,j}}: = \left\{ {x \in \gf_{{q}}^\# :\chi (x) = i,\chi (x + 1) = j} \right\} $, where $ i,j\in \{\pm1\} $.
	\item $ {S_{1,1}} \cup {S_{ - 1, - 1}} \cup {S_{1, - 1}} \cup {S_{ - 1,1}} = \gf_{{q}}^\# $.
\end{itemize}

Secondly, we introduce some lemmas which will be used in the sequel. To determine the greatest common divisor of integers, the following lemma plays an important role in the rest of this paper.

\begin{lemma}
	let $ p,k,n $ be integers greater than or equal to $ 1 $. Then
	\[\gcd \left( {{p^k} + 1,{p^n} - 1} \right) = \left\{ \begin{aligned}{}
		\frac{{{2^{\gcd \left( {2k,n} \right)}} - 1}}{{{2^{\gcd \left( {k,n} \right)}} - 1}},~~~&\mathrm{if}~{p} = 2,\\
		2,~~~~~~~~~~~~~~&\mathrm{if}~\frac{n}{{\gcd (n,k)}}~\mathrm{is}~\mathrm{odd},\\
		{p^{\gcd \left( {k,n} \right)}} + 1,~~~~&\mathrm{if}~\frac{n}{{\gcd \left( {n,k} \right)}}~\mathrm{is}~\mathrm{even}.
	\end{aligned} \right.\]
\end{lemma}

The following lemma can be used in solving equations over finite fields.
\begin{lemma}\cite{FF}\label{quadraticequationsolution}
	Let $ n $ be a positive integer. We have:

	\noindent ({\rm i}) The equation $ x^{2}+ax+b=0 $, with $ a, b \in \gf_{2^{n}} $, $ a \ne 0 $, has two solutions in $ \gf_{2^{n}} $ if $ \tr_{n}(\frac{b}{a^{2}})=0 $, and zero solutions otherwise.
	
	\noindent ({\rm ii}) The equation $ x^{2}+ax+b=0 $, with $ a, b \in \gf_{p^{n}} $, $ p $  odd, has (two, respectively, one) solutions in $ \gf_{p^{n}} $ if and only if the discriminant $ a^{2}-4b $ is a (nonzero, respectively, zero) square in $ \gf_{p^{n}} $.
 \end{lemma}
At last, we introduce the following result on the quadratic multiplicative character sums.
\begin{lemma} \cite[Theorem 5.48]{FF}\label{charactersumquadratic} Let $f(x)=a_2x^2+a_1x+a_0\in\gf_q[x]$ with $q$ odd and $a_2\neq0$. Put $d=a^2_1-4a_0a_2$ and let $\chi$ be the quadratic character of $\gf_q$. Then
	\begin{eqnarray*}
		\sum_{x\in\gf_q}\chi(f(x))=\left\{
		\begin{array}{lllll}
			-\chi(a_2), ~~~~~~~\mathrm{if}~d\neq0.\\
			(q-1)\chi(a_2), ~\mathrm{if}~d=0.
		\end{array} \right.\ \
	\end{eqnarray*}
 \end{lemma}
\section{The properties of the $c$-differential spectrum}
The usual differential spectrum of a power function  satisfies several identities (see  \cite{BCC}). It is natural to wonder how it behaves with respect to the $c$-differential spectrum of a power function. In this section, we give the following identities of the $c$-differential spectrum.

\begin{theorem}\label{properties1}Let $ F\left( x \right) = {x^d} $ be a power function over $ \gf_{q} $ with $c$-differential uniformity $_c{\Delta _F}$ for some $1\neq c\in\gf_{q}$. Recall that $ {{}_c\omega _i} = \# \left\{ {b \in \gf_{q} : {{}_c{\Delta _F}\left(1, b \right) = i}.} \right\} $ for each $ 0 \le i \le {}_c{\Delta _F} $. We have 
	\begin{equation}\label{omegaiomega}
		\sum_{i=0}^{{}_c{\Delta _F}}{}_c\omega_i=\sum_{i=0}^{{}_c{\Delta _F}}i\cdot{}_c\omega_i=q.
	\end{equation} 	
	Moreover, we have 
	\begin{equation}\label{i^2omega_ieqution}
		\sum\limits_{i = 0}^{{}_c{\Delta _F}} {{i^2}\cdot {}_c{\omega _i}}  = \frac{{{\rm{ }}{{}_cN_4} - 1}}{{q} - 1}-\gcd(d,q-1) ,
	\end{equation}
	where
  \begin{equation}\label{equationsystem}
  	{}_cN_{4}=\# \left\{ {\left( {{x_1},{x_2},{x_3},{x_4}} \right) \in (\gf_{q})^4 : {\Bigg\{ \begin{array}{l}
  				{x_1} - {x_2} + {x_3} - {x_4} = 0\\
  				x_1^d - cx_2^d + cx_3^d - x_4^d = 0
  		\end{array} } } \right\}. 
  \end{equation} 
\end{theorem}
	\begin{proof}
		According to the definition of $ {}_c\omega_{i} $, we have

	\[\sum\limits_{i = 0}^{{}_c{\Delta _F}} {{{}_c\omega _i}} = \sum\limits_{i = 0}^{{}_c{\Delta _F}} {\# \left\{ {b \in \gf_{q} : {{}_c{\Delta _F}\left(1, b \right) = i} } \right\}}
	= \sum\limits_{b \in \gf_{q}} 1  = q.\]
and 
\begin{align*}
	\sum\limits_{i = 0}^{{}_c{\Delta _F}} {i\cdot{{}_c\omega _i}}  &=\sum\limits_{i = 0}^{{}_c{\Delta _F}} {i \cdot\# \left\{ {b \in \gf_{q} : {{}_c{\Delta _F}(1, b ) = i}} \right\}}\\
	&=\sum_{b\in\gf_{q}}\#\{x\in\gf_{q}:(x+1)^d-cx^d=b\}\\	
	&= \sum\limits_{x \in \gf_q } 1 \\
	&= q.
\end{align*}
Next we prove the last statement. For $c\neq1$, we define \[{}_cn({\alpha,\beta }) = \# \Bigg\{ {({{x},{y}}) \in (\gf_{q})^2: {\bigg\{ \begin{array}{l}
					{x} - {y} = \alpha \\
					x^d - cy^d = \beta 
			\end{array} }} \Bigg\}. \]
		It is obvious that
		$$ {{}_cN_4} = \sum\limits_{\alpha ,\beta  \in {\gf_{q}}} {{(_cn}( {\alpha ,\beta }))^2}. $$
		Moreover, one can easily check that $ {}_cn(0,0)=1 $. For $\beta\neq 0$, $_cn(0,\beta)$ is equal to the number of solutions $y\in\gf_{q}$ of the equation $y^d=\frac{\beta}{1-c}$. Let  $ \gamma $ be a primitive element of $ \gf_{q}^*$ and $\frac{\beta}{1-c}=\gamma^{k}$ for some integer $k$. Let $e=\gcd(d,q-1)$. Then $_cn(0,\beta)=e$ if $e\mid k$ and $_cn(0,\beta)=0$ otherwise. For $ \alpha,\beta \ne 0 $, we have  $ {}_cn(\alpha,\beta) = {}_cn(1,\frac{\beta }{\alpha ^d}) $. We immediately obtain the following.
		\[\begin{aligned}{}
			{{}_cN_4} &=\sum\limits_{\alpha ,\beta  \in {\gf_{q}}} {{(_cn}( {\alpha ,\beta } ))^2}\\
			&= {({}_cn}({0,0}))^2 +  \sum\limits_{\beta  \in \gf_{q}^*} ({{}_cn}({0,\beta }))^2 + \sum\limits_{\alpha  \in \gf_{q}^*} { {\sum\limits_{\beta  \in {\gf_{q}}} {({{}_cn}\left( {\alpha, \beta} \right))^2} }} \\
			&= 1 +  \sum_{1\leq k \leq q-1, e\mid k} e^2  + \sum\limits_{\alpha  \in \gf^*_{q}} { {\sum\limits_{\beta  \in {\gf_{q}}} {({{}_cn}( {1,\frac{\beta }{{{\alpha ^d}}}}))^2} }} \\
			&= 1 + \frac{{q}-1}{e}{{e}^2} + \sum\limits_{\alpha  \in \gf_{q}^*} { {\sum\limits_{b \in {\gf_{q}}} {({{}_cn}( {1,b}))^2} }} \\
			&= 1 + \big(q-1 \big){e} + \sum\limits_{\alpha  \in \gf_{q}^*} {{\sum\limits_{b \in {\gf_{q}}} {{{\big( {{\delta _c}(b)} \big)}^2}} }} \\
			&= 1 + \big(q-1 \big){{ {e} }} + (q - 1)\sum\limits_{i = 0}^{{}_c{\Delta _F}} {{i^2}\cdot{{}_c\omega _i}}. 
		\end{aligned}\]
		The fourth identity holds since for fixed $a\neq 0$, when $\beta$ runs through $\gf_{q}$, so does $\frac{\beta}{\alpha^d}$. We complete the proof.
	\end{proof}
	
\section{The $ c $-differential spectrum of some power functions}
In this section, we investigate the $c$-differential spectra of some power functions with low $c$-differential uniformity.
\subsection{The $c$-differential spectrum of the inverse function}
Since there has been quite a bit of effort to investigate the inverse function over $ \gf_{2^{n}} $ as it is relevant in Rijndael and Advance Encryption Standard, it is natural to wonder how it behaves with respect to the $ c $-differential uniformity. In \cite{2020CEP}, the $ c $-differential uniformity of the inverse function $ x\in \gf_{p^{n}},x\mapsto x^{p^{n}-2}  $ has been studied. We have the following lemmas.
\begin{lemma}[\cite{2020CEP}] 
	Let $ n $ be a positive integer, $ 1\ne c \in \gf_{2^{n}} $ and $ F:{\gf_{{2^n}}} \to {\gf_{{2^n}}} $ be the inverse function defined by $ F(x)=x^{2^{n}-2} $. We have:
	
	\noindent ({\rm i}) If $ c=0 $, then $ F $ is  P$ c $N ;
	
	\noindent ({\rm ii}) If $ c\ne 0 $ and $ \tr_n(c)=\tr_n(\frac{1}{c})=1 $, the $ c $-differential uniformity of $ F $ is $ 2 $;
	
	\noindent ({\rm iii}) If $ c\ne 0 $ and $ \tr_n(c)=0$ or $\tr_n(\frac{1}{c})=0 $, the $ c $-differential uniformity of $ F $ is $ 3 $.
	
\end{lemma}		

\begin{lemma}[\cite{2020CEP}]
	Let $p$ be an odd prime. $ n\ge 1 $ be a positive integer, $ 1\ne c \in \gf_{p^{n}} $ and $ F:{\gf_{{p^n}}} \to {\gf_{{p^n}}} $ be the inverse function defined by $ F(x)=x^{p^{n}-2} $. We have:
	
	\noindent ({\rm i}) If $ c=0 $, then $ F $ is P$ c $N .
	
	\noindent ({\rm ii}) If $ c\neq 0, 4, 4^{-1}$, $\chi(c^2-4c)=1$, or $ \chi(1-4c)=1 $, the $ c $-differential uniformity of $ F $ is $ 3 $.
	
	\noindent ({\rm iii}) If $ c = 4, 4^{-1}$, the $ c $-differential uniformity of $ F $ is $ 2 $.
	
	\noindent ({\rm iv}) If $ c \neq 0 $, $\chi(c^{2}-4c)=\chi(1-4c)=-1$, the $ c $-differential uniformity of $ F $ is $ 2 $.
	
\end{lemma}	

The $ c $-differential spectrum of the inverse function is given in the following theorem. 	
	
\begin{theorem}\label{Inversep=2APcN}
	Let $ F(x)=x^{2^{n}-2} $ be a power function over $ \gf_{{2^n}}$. When $ 0,1\ne c \in \gf_{2^{n}}$, the inverse function is AP$c$N with  $ c $-differential spectrum 
	$$
	\mathbb S=\{{}_c\omega_0=2^{n-1}-2,{}_c\omega_1=4,{}_c\omega_2=2^{n-1}-2\}
	$$ if  $ \tr_n(c)=\tr_n(\frac{1}{c})=1 $.
	Moreover, the inverse function is differentially $ (c,3) $-uniform with  $ c $-differential spectrum
	$$
	\mathbb S=\{{}_c\omega_0=2^{n-1}-1,{}_c\omega_1=3,{}_c\omega_2=2^{n-1}-3,{}_c\omega_3=1\}
	$$
	if $\tr_n(c)=1$, $\tr_n(\frac{1}{c})=0 $ or $\tr_n(c)=0$, $\tr_n(\frac{1}{c})=1$, and is differentially $ (c,3) $-uniform with $ c $-differential spectrum
	$$
	\mathbb S=\{{}_c\omega_0=2^{n-1},{}_c\omega_1=2,{}_c\omega_2=2^{n-1}-4,{}_c\omega_3=2\}
	$$
	if $\tr_n(c)=\tr_n(\frac{1}{c})=0$.
\end{theorem}
\begin{proof}
	For $ b \in \gf_{{2^n}} $, we consider the $c$-differential equation
	\begin{equation}\label{2-c}
		(x+1)^{2^n-2} + cx^{2^n-2}=b.
	\end{equation}
	In order to calculate the $ c $-differential spectrum, we first determine the values of ${}_c\omega_1$ and ${}_c\omega_3$. 
	
	{\it Case 1:} $b=0$. (\ref{2-c}) has a unique solution $x=\frac{c}{c+1}$. 
	
	{\it Case 2:} $ b=1 $. Then $ x=0 $ is a solution of (\ref{2-c}), and $x=1$ is not a solution of (\ref{2-c}) since $c\neq 1$. Assume  that $ x\neq0,1 $, by multiplying $ x(x+1) $ on both sides of (\ref{2-c}), we get $ x^{2}+cx+c=0 $. By Lemma \ref{quadraticequationsolution} (i), this equation has two solutions iff $ \tr_{n}(\frac{1}{c})=0 $. Consequently, $\delta_c(1)=3$ if  $ \tr_{n}(\frac{1}{c})=0 $, and $\delta_c(1)=1$ if  $ \tr_{n}(\frac{1}{c})=1 $.
	
	{\it Case 3:} $ b=c$. In this case, $ x=1 $ is a solution of (\ref{2-c}), while $ x=0 $ is not a solution. Next we assume that $ x\neq0,1 $. Similarly, we get $ x^{2}+c^{-1}x+1=0 $. By Lemma \ref{quadraticequationsolution} (i), this equation has two solutions iff $ \tr_{n}(c)=0 $. Consequently, $\delta_c(c)=3$ if  $ \tr_{n}(c)=0 $, and $\delta_c(c)=1$ if  $ \tr_{n}(c)=1 $.
	
	{\it Case 4:} $ b\neq 0,1,c $. In this case, $ x\neq 0, 1 $. Multiplying  $ x(x+1) $ on both sides of (\ref{2-c}), we obtain 	
	\begin{equation}\label{2-main}
		bx^2+(1+b+c)x+c=0.
	\end{equation}
	When $b=1+c$, (\ref{2-main}) has unique solution $x=(\frac{c}{c+1})^{2^{n-1}}$.
	By Lemma \ref{quadraticequationsolution} (i), (\ref{2-main}) has 0 or 2 solutions if $1+b+c \neq 0$. We have $\delta_c(1+c)=1$, and $\delta_c(b)=0$ or $2$ when $b\neq 0,1, c, 1+c $.

	Summarizing the above, we obtain that 
	\[{}_c\omega_1 = \left\{ \begin{aligned}{}
		4,~~~&\mathrm{if}~\tr_n(c)=\tr_n(\frac{1}{c})=1 ,\\
		3,~~~&\mathrm{if}~\tr_n(c)=1, \tr_n(\frac{1}{c})=0~ \mathrm{or}~\tr_n(c)=0, \tr_n(\frac{1}{c})=1 ,\\
		2,~~~&\mathrm{if}~\tr_n(c)=\tr_n(\frac{1}{c})=0,
	\end{aligned} \right.\]
	and 
	\[{}_c\omega_3= \left\{ \begin{aligned}{}
		0,~~~&\mathrm{if}~\tr_n(c)=\tr_n(\frac{1}{c})=1 ,\\
		1,~~~&\mathrm{if}~\tr_n(c)=1, \tr_n(\frac{1}{c})=0~ \mathrm{or}~\tr_n(c)=0, \tr_n(\frac{1}{c})=1 ,\\
		2,~~~&\mathrm{if}~\tr_n(c)=\tr_n(\frac{1}{c})=0.
	\end{aligned} \right.\]
	By solving the equation system (\ref{omegaiomega}), the desired result follows.

\end{proof}

Next, we calculate the $ c $-differential spectrum of the inverse function over $\gf_{p^n}$ when $p$ is odd. We have the following results.
\begin{theorem}\label{thm-oddAPcN1}
	Let $p$ be any odd prime. $ F(x)=x^{p^{n}-2} $ be a power function over $ \gf_{p^{n}}$. 
	When $ 0,1,4,4^{-1}\ne c \in \gf_{p^{n}}$, the $ c $-differential spectrum of it is given as the following six cases:
	
	\noindent ({\rm i}) $\mathbb S=\{{}_c\omega_0=\frac{p^n-3}{2},{}_c\omega_1=3,{}_c\omega_2=\frac{p^n-3}{2}\}$ if $\chi(c^2-4c)=\chi(1-4c)=-1$ and $\chi(c)=-1$;
	
	\noindent ({\rm ii}) $\mathbb S=\{{}_c\omega_0=\frac{p^n-5}{2},{}_c\omega_1=5,{}_c\omega_2=\frac{p^n-5}{2}\}$ if $\chi(c^2-4c)=\chi(1-4c)=-1$ and $\chi(c)=1$;
	
	\noindent ({\rm iii}) $\mathbb S=\{{}_c\omega_0=\frac{p^n-1}{2},{}_c\omega_1=2,{}_c\omega_2=\frac{p^n-5}{2},{}_c\omega_3=1\}$ if $\chi(c^2-4c)\chi(1-4c)=-1$ and $\chi(c)=-1$;
	
	\noindent ({\rm iv}) $\mathbb S=\{{}_c\omega_0=\frac{p^n-3}{2},{}_c\omega_1=4,{}_c\omega_2=\frac{p^n-7}{2},{}_c\omega_3=1\}$ if $\chi(c^2-4c)\chi(1-4c)=-1$ and $\chi(c)=1$;
	
	\noindent ({\rm v}) $\mathbb S=\{{}_c\omega_0=\frac{p^n+1}{2},{}_c\omega_1=1,{}_c\omega_2=\frac{p^n-7}{2},{}_c\omega_3=2\}$ if $\chi(c^2-4c)=\chi(1-4c)=1$ and $\chi(c)=-1$;
	
	\noindent ({\rm vi}) $\mathbb S=\{{}_c\omega_0=\frac{p^n-1}{2},{}_c\omega_1=3,{}_c\omega_2=\frac{p^n-9}{2},{}_c\omega_3=2\}$ if $\chi(c^2-4c)=\chi(1-4c)=1$ and $\chi(c)=1$.
	
\end{theorem}

\begin{proof}
	For $ b \in \gf_{{p^n}} $, we consider the $ c $-diffferential equation
	\begin{equation}\label{poddc}
		(x+1)^{p^n-2}-cx^{p^n-2}=b.
	\end{equation}	
	Similarly, we determine the values of ${}_c\omega_1$ and ${}_c\omega_3$.	
	
	{\it Case 1: } $b=0$. It can be easily seen that $x=\frac{c}{1-c}$ is the unique solution of (\ref{poddc}), i.e. $\delta_c(0)=1$.
	
	{\it Case 2: } $ b=1$. Then $ x=0 $ is a solution of (\ref{poddc}), and $ x=-1 $ is not a solution of (\ref{poddc}). Assume that $ x\neq0,-1 $, by multiplying $ x(x+1) $ on both sides of (\ref{poddc}), we get $ x^{2}+cx+c=0 $. The discriminant of this quadratic equation is $c^{2}-4c $, which is not zero since $c\neq 0,4$. By Lemma \ref{quadraticequationsolution}(ii), this equation has two solutions if $\chi(c^2-4c)=1$, and no solution if $\chi(c^2-4c)=-1$. We obtain $\delta_c(1)=3$ if $\chi(c^2-4c)=1$, $\delta_c(1)=1$ if $\chi(c^2-4c)=- 	1$.
	
	{\it Case 3: } $ b=c$. Then $ x=-1 $ is a solution of (\ref{poddc}), while $ x=0 $ is not a solution. Next we assume that $ x\neq0,-1 $. By multiplying $ x(x+1) $ on both sides of  (\ref{poddc}), we get $ cx^{2}+(2c-1)x+c=0 $. The discriminant of this quadratic equation is $1-4c $, which is not zero since $c\neq 0,4^{-1}$. By Lemma \ref{quadraticequationsolution}(ii), this equation has two solutions if $ \chi(1-4c)=1 $, 
	and no solution if $ \chi(1-4c)=-1 $. We obtain $\delta_c(c)=1$ if $ \chi(1-4c)=-1 $ and $\delta_c(c)=3$ if $ \chi(1-4c)=1 $.
	
	{\it Case 4: } Let $ b\neq 0, 1,c $. Then $ x\neq 0, 1 $. By multiplying $ x(x+1) $ on both sides of (\ref{poddc}), we obtain 	
	\begin{equation}\label{p-main}
		bx^{2}+(b+c-1)x+c=0.
	\end{equation}
	The equation (\ref{p-main}) has a unique solution if and only if the discriminant
	$(b+c-1)^2-4bc=0$. That is,  
	\[b^2-2(c+1)b+(c-1)^2=0.\]
	The above equation on $b$ is quadratic, and the discriminant is $ 16c $. We assert that, if $\chi(c)=1$, there exist two $b$'s such that (\ref{p-main}) has one solution, and if $\chi(c)=-1$, such $b$ do not exist. 
	
	By discussions as above, we summarize the values of $_c{\omega_1}$ and $_c{\omega_3}$ in the following table. By solving the equation system (\ref{omegaiomega}), the desired result follows.
	\begin{table}[!htbp]
		\begin{center}
			\begin{minipage}{174pt}
				\caption{The values of $_c{\omega_1}$ and $_c{\omega_3}$}\label{table-2}
				\begin{tabular}{@{}lllll@{}}
					\toprule
					 $_c{\omega_1}$ & $_c{\omega_3}$ &$\chi(c^{2}-4c)$&$\chi(1-4c)$&$\chi(c)$ \\
					\midrule
					 $3$ & $2$ &$ 1 $& $ 1 $& $ 1 $\\
					 $1$ & $2$ &$ 1 $& $ 1 $& $ -1 $\\
					 $2$ & $1$ &$ 1 $& $ -1 $& $ -1 $\\
					 $4$ & $1$ &$ 1 $& $ -1 $& $ 1 $\\
					 $4$ & $1$ &$ -1 $& $ 1 $& $ 1 $\\
					 $2$ & $1$ &$ -1 $& $ 1 $& $ -1 $\\
					 $3$ & $0$ &$ -1 $& $ -1 $& $ -1 $\\
					 $5$ & $0$ &$ -1 $& $ -1 $& $ 1 $\\
					\botrule
				\end{tabular}
			\end{minipage}
		\end{center}
	\end{table}
	
\end{proof}	
	
\subsection{The $(-1)$-differential spectra of some power functions}

The case that $c=-1$ is a special case. Sometimes the $(-1)$-differential uniformity of a power function is lower than its $c$-differential uniformity for other $c$. In this subsection, we compute the $(-1)$-differential spectra of some power functions. We begin this subsection with a simple result. 

In \cite{2021SYZ}, the authors reported that the power function $x^{\frac{3^{n}+3}{2}} $ over $ \gf_{3^{n}} $ is AP$c$N if $ n $ is even. We have the following theorem.
\begin{theorem}\label{p=3-niseven}
	Let $ G\left( x \right) = {x^{\frac{{{3^n} +3}}{2}}} $ on $ \gf_{3^{n}} $, $ n\ge 2 $. The $ (-1) $-differential spectrum of $ G(x) $ is 
	\[\mathbb{S} = \left\{ {{{}_{-1}\omega _0} = \frac{{{3^n} - 1}}{2},{{}_{-1}\omega _1} = 1,{{}_{-1}\omega _2} = \frac{{{3^n} - 1}}{2}} \right\}.\]  
\end{theorem}

\begin{proof}
	To determine the $ (-1) $-differential spectrum of $x^{\frac{{{3^n} + 3}}{2}}$, we consider the equation
	\begin{equation}\label{3-c}
		\Delta_{-1}(x)={\left( {x + 1} \right)^{\frac{{{3^n} + 3}}{2}}} + {x^{\frac{{{3^n} + 3}}{2}}} = b
	\end{equation}
	for $ b\in \gf_{{3^n}} $.
	Note that  $ x $ is a solution of (\ref{3-c}) if and only if $ -x-1 $ is a solution of (\ref{3-c}). When $x=-x-1$, then $x=1$ and the corresponding $b=-1$ satisfies that $\delta_{-1}(-1)$ is odd. Then  $\delta_{-1}(-1)=1$ since $F(x)$ is AP$c$N. We get $ _{-1}\omega _1=1 $. Then by solving the equation (\ref{omegaiomega}), we complete the proof.
\end{proof}

It was shown that $ x^{3^{n}-3} $ is an APN power mapping of $ \gf_{3^{n}} $ when $ n>1 $ is odd \cite{HRS}, is differentially $4$-uniform when $ n = 2\left( {\bmod 4} \right)$ \cite{2020XZL}, and is differentially $5$-uniform when $ n = 0\left( {\bmod 4} \right) $ \cite{2020XZL}. Moreover, \cite{2021SYZ} showed that this power function has low $ (-1) $-differential uniformity. In the following theorem, we determine the $ (-1) $-differential spectrum of $x^{3^n-3}$ over $\gf_{3^n}$.

\begin{theorem}\label{p=3d=-2} Let $ H\left( x \right) = {x^{3^{n}-3}} $ be a power function over $ {\gf_{{3^n}}} $. When $n = 0\left( {\bmod 4} \right) $, $ H $ is differentially $ (-1,6) $-uniform with $ (-1) $-differential spectrum  
	\[\begin{gathered}
		\mathbb{S} = \left\{ {{{}_{-1}\omega _0} = \frac{{5 \cdot {3^n} - 3}}{8},{{}_{-1}\omega _1} = 1,{{}_{-1}\omega _2} = \frac{{{3^n} + 3}}{4},{{}_{-1}\omega _4} = \frac{{{3^n} - 17}}{8}, 
	{{}_{-1}\omega _6} = 1} \right\}. \hfill \\ 
	\end{gathered} \]
	When $n = 2\left( {\bmod 4} \right) $, $ H $ is differentially $ (-1,4) $-uniform with $ (-1) $-differential spectrum \[\mathbb{S} = \left\{ {{{}_{-1}\omega _0} = \frac{{5 \cdot {3^n} - 13}}{8},{{}_{-1}\omega _1} = 1,{{}_{-1}\omega _2} = \frac{{{3^n} + 7}}{4},{{}_{-1}\omega _4} = \frac{{{3^n} - 9}}{8}} \right\}.\]
	When $n = 1,3\left( {\bmod 4} \right) $, $ H $ is differentially $ (-1,4) $-uniform with $ (-1) $-differential spectrum \[\mathbb{S} = \left\{ {{{}_{-1}\omega _0} = \frac{{5 \cdot {3^n} - 7}}{8},{{}_{-1}\omega _1} = 1,{{}_{-1}\omega _2} = \frac{{{3^n} + 1}}{4}, {{}_{-1}\omega _4} = \frac{{{3^n} - 3}}{8}{\rm{ }}} \right\}.\]	
\end{theorem}

\begin{proof} 
	To determine the $ (-1) $-differential spectrum of this power function, we consider the equation
	\begin{equation}\label{1}
		\Delta_{-1}(x)=(x+1)^{3^n-3} + x^{3^n-3}=b
	\end{equation}
	for $ b \in \gf_{{3^n}} $. Firstly we consider $ b=1$. Note that $ x=0 $ and $ x=-1 $ are both solutions of (\ref{1}). For $ x \ne 0,-1 $, we multiply both sides of  (\ref{1}) by $ x^{2}(x+1)^{2} $ and obtain
	\begin{equation}\label{2}
		x^{4}-x^{3}-x^{2}+x-1=0.
	\end{equation}
	Let $ y=x-1 $, then $ y\ne-1,1 $, and (\ref{2}) becomes
	\begin{equation}\label{4}
		y^{4}-y^{2}-1=0.
	\end{equation}
	Next we prove that the solutions of (\ref{4}) are all in $\gf_{3^4}\setminus\gf_{3^2}$. Obviously, the solutions of (\ref{4}) are in $\gf_{3^4}$. If $y\in \gf_{3^2}$ is a solution of (\ref{4}), then $y\neq 0$ and $y^4=\pm 1$. If $y^4=1$ (respectively, $y^4=-1$), we can obtain $y^2=0$ by (\ref{4}) (respectively, $y^2=1$), which is a contradiction. By the discussion as above, we conclude that $\delta_{-1}(1)=6$ when $n = 0\left( {\bmod 4} \right)$ and $\delta_{-1}(1)=2$ otherwise.
	
	For $b\neq 1$, we know that $x=0$ and $x=-1$ are not solutions of (\ref{1}). When $b=0$, the equation (\ref{1}) becomes
	\begin{equation}\label{0}
		x^2+x-1=0.
	\end{equation}
	It can be similarly proved that the solutions of (\ref{0}) are in $\gf_{3^2}\setminus \gf_3$. We conclude that $\delta_{-1}(0)=2$ if $n$ is even, and  $\delta_{-1}(0)=0$ if $n$ is odd.

	If $  b \ne 0,1 $, then $x\neq0,1$. The differential equation (\ref{1}) can be deduced as
	\begin{equation}\label{2.5}
		bx^4-bx^3+(b+1)x^2+x-1=0.
	\end{equation}
	Let $y=x-1$, then (\ref{2.5}) becomes
	\begin{equation}\label{3}
		by^{4}+(b+1)y^{2}+b+1=0.
	\end{equation}  
	We consider the solutions of equation (\ref{3}). Let $ z=y^{2} $, then we have
	\begin{equation}\label{6}
		bz^{2}+(b+1)z+b+1=0.
	\end{equation}
	Note that (\ref{6}) is a quadratic equation on $z$, and the discriminant of this quadratic equation is $\Delta=b+1 $. We discuss in the following three cases.
	
	\noindent Case ({\rm i}) $b=-1$. We get $ z=0 $, then $ y=0 $ and $ x=1 $ consequently. So $ \delta_{-1}(-1)=1 $.
	
	\noindent Case ({\rm ii}) $\chi \left( {b + 1} \right) = 1$. Herein, $\chi$ denotes the quadratic character on $\gf_{3^n}$.
	We know that (\ref{6}) has two distinct nonzero solutions, namely $ z_{1}$ and $z_{2} $. Now we need to determine whether $ y^{2}=z_i $ ($i=1, 2$) has solutions over $ \gf_{3^{n}} $. If $ \chi \left( {z_{1}} \right) =\chi \left( {z_{2}} \right) =1 $, then (\ref{3}) has four solutions. If $ \chi \left( {{z_1}{z_2}} \right) =  - 1 $, (\ref{3}) has two solutions. If $ \chi \left( {z_{1}} \right) =\chi \left( {z_{2}} \right) =-1 $, (\ref{3}) has no solution. 
	
	\noindent Case ({\rm iii}) $\chi \left( {b + 1} \right) = -1$. Obviously, (\ref{3}) has no solution.

	Summarizing the discussion as above, (\ref{1}) cannot have $3$ or $5$ solutions. We have, $_{-1}\omega_3=_{-1}\omega_5=0$. Moreover, $_{-1}\omega_1=1$, $_{-1}\omega_6=1$ if $n = 0\left( {\bmod 4} \right) $ and $_{-1}\omega_6=0$ otherwise. By (\ref{omegaiomega}), we know that $_{-1}\omega_0$, $_{-1}\omega_2$ and $_{-1}\omega_4$ satisfy
	\begin{align}\label{omega6d=-2}
		\left\{
		\begin{array}{ll}
			_{ - 1}{\omega _0}{ + _{ - 1}}{\omega _2}{ + _{ - 1}}{\omega _4} &= {3^n} - 1{ - _{ - 1}}{\omega _6},\\
			{2_{ - 1}}{\omega _2} + {4_{ - 1}}{\omega _4} &= {3^n} - 1 - {6_{ - 1}}{\omega _6}.
		\end{array}
		\right.
	\end{align}
	
	Denote by $N=\#\{b\in\gf_{3^n}\setminus \gf_3~:~\chi(b+1)=1, \chi(b)=-1\}$. Then
	$_{-1}\omega_2=N+1$ when $ n = 0,1,3\left( {\bmod 4} \right) $, $_{-1}\omega_2=N+2$ when $ n = 2\left( {\bmod 4} \right) $. By Lemma \ref{charactersumquadratic}, we have
	\begin{align*}
		N&=\frac{1}{4}\sum\limits_{b\in\gf_{3^n}\setminus\gf_3} {(\chi(b+1)+ 1 )}{( {1-\chi (b) }) }\\
		&=\frac{1}{4}(\sum\limits_{b\in\gf_{3^n}} {(\chi(b+1)+ 1 )}{( {1-\chi (b) }) -3+\chi(-1))}\\
		&=\frac{1}{4}(-\sum\limits_{b\in\gf_{3^n}}\chi(b(b+1))+ \sum\limits_{b\in\gf_{3^n}}\chi(b+1)-\sum\limits_{b\in\gf_{3^n}}\chi(b)+\sum\limits_{b\in\gf_{3^n}}1-3+\chi(-1))\\
		&=\frac{3^n-2+\chi(-1)}{4}.
	\end{align*}
	Then by solving (\ref{omega6d=-2}), the desired result follows. 
\end{proof}

In \cite{2021SYZ}, the authors studied the $ (-1) $-differential uniformity of the power mapping $ x^{\frac{p^{k}+1}{2}} $, where $p$ is an odd prime. In the following two theorems, we determine the $ (-1) $-differential spectrum of $ x^{\frac{p^{k}+1}{2}} $ with some conditions.
\begin{theorem}\label{poddp+1/2-1}
	Let $ F(x)=x^{\frac{p^{k}+1}{2}} $ be a power function over $ \gf_{p^{n}}$, where $ p \equiv 1\left( {\bmod~4} \right) $, $ \gcd(n,k)=1 $, and $ \frac{2n}{\gcd(2n,k)} $ is even. The $ (-1) $-differential spectrum of $F(x)$ is given by 
	\[\mathbb{S} = \left\{ {{}_{-1}{\omega _0} = \frac{{({p^n} + 1)(p-1)}}{2(p+1)},{}_{-1}{\omega _1} = \frac{{{p^n} - 3}}{2},{}_{-1}{\omega _{\frac{{p + 3}}{4}}} = 2,{}_{-1}{\omega _{\frac{{p + 1}}{2}}} = \frac{{{p^n} -p}}{{p + 1}} } \right\} \]
	when $ n $ is odd, and is given by
	\[\mathbb{S} = \left\{ {{}_{-1}{\omega _0} = \frac{{(p^{n}-1)(p-1)}}{2(p+1)},{}_{-1}{\omega _1} = \frac{{{p^n} - 1}}{2},{}_{-1}{\omega _{\frac{{p + 3}}{4}}} = 2,{}_{-1}{\omega _{\frac{{p + 1}}{2}}} = \frac{{{p^n} - p-2}}{{p + 1}} } \right\} \]
	when $ n $ is even.
\end{theorem}

\begin{proof}
	Since $ \frac{2n}{\gcd(2n,k)} $ is even and $\gcd(n,k)=1$, then $ k $ is odd. We consider the $(-1)$-differential function $ \Delta_{-1} \left( x \right) = {\left( {x + 1} \right)^d} + {x^d} $ on $ \gf_{p^{n}} $. There exist $ \alpha ,\beta  \in \gf_{p^{2n}}  $ such that $ x + 1 = {\alpha ^2}$, $ x = {\beta ^2} $. Let $ \alpha - \beta = \theta \in \gf_{{p^{2n}}}^*  $, then $ \alpha  + \beta  = {\theta ^{ - 1}} $, $ \alpha  = \frac{1}{2}(\theta  + {\theta ^{ - 1}}) $,
	$ \beta  =  - \frac{1}{2}(\theta  - {\theta ^{ - 1}}) $, and $ x = \frac{1}{4}{\left( {\theta  - {\theta ^{ - 1}}} \right)^2} $. We can obtain that $ \theta^{2(p^{n}+1)}=1 $ or $ \theta^{2(p^{n}-1)}=1 $ since $ x \in \gf_{p^{n}} $. 	Note that $ x= 0 $ if and only if $ \theta^{2}=1 $, and $ x= -1 $ if and only if $ \theta^{2}=-1 $.  For $\theta^2\neq \pm1$, we can check that $\pm\theta,\pm\theta^{-1}$ are pairwise distinct. We mention that for $ x\neq 0,-1  $, each $ x $ corresponds to four distinct $ \theta  $'s $ ( \pm \theta,\pm \theta^{-1}) $. Moreover, we have 
	\[\begin{aligned}
		\Delta_{-1} \left( x \right) &= {\alpha ^{{p^k} + 1}} + {\beta ^{{p^k} + 1}}\\
		&= \frac{1}{4}{\left( {\theta  + {\theta ^{ - 1}}} \right)^{{p^k} + 1}} + \frac{1}{4}{\left( {\theta  - {\theta ^{ - 1}}} \right)^{{p^k} + 1}}\\
		&= \frac{1}{2}({\theta ^{{p^k} + 1}} + {\theta ^{ - {p^k} - 1}}).
	\end{aligned}\]	
	Although we can choose different $ \theta $, $ \Delta_{-1} \left( x \right)=\frac{1}{2}({\theta ^{{p^k} + 1}} + {\theta ^{ - {p^k} - 1}}) $ always holds, no matter which $ \theta $ is chosen. To determine the $ (-1) $-differential spectrum of this function, we need to investigate the image of the $\Delta_{-1}(x)$. Let $\mathrm{Im}(\Delta_{-1})\mid_{S}$ denote the image set of the differential function $\Delta_{-1}(x)$ restricted on some set $S$, i.e.,
	$\mathrm{Im}(\Delta_{-1})\mid_{S}=\{\Delta_{-1}(x) : x=\frac{1}{4}(\theta-\theta^{-1})^2, \theta\in S\}$. 
	We define
	\[ {S_1} = \left\{ {\theta : {{\theta ^{2\left( {{p^n} + 1} \right)}} = 1}} \right\}, {S_2} = \left\{ {\theta : {{\theta ^{2\left( {{p^n} - 1} \right)}} = 1}} \right\}.\]

	Then we have 
	\[{S_1} = \left\{ {{\gamma ^i}:{i = \frac{{{p^n} - 1}}{2}j,0 \le j \le 2{p^n} + 1}} \right\}, {S_2} = \left\{ {{\gamma ^i}: {i = \frac{{{p^n} + 1}}{2}j,0 \le j \le 2{p^n} - 3} } \right\}, \]
	where $\gamma$ is a generator of  $ \gf_{{p^{2n}}}^*$. Define
	\[{C_i} = \left\{ {{\theta ^{{p^k} + 1}}: {\theta  \in {S_i}} } \right\}, i=1, 2.\]
	Then, $ {C_1} = \left\langle {{\gamma ^{\gcd (\frac{{({p^n} - 1)({p^k} + 1)}}{2},{p^{2n}} - 1)}}} \right\rangle  $, $ {C_2} = \left\langle {{\gamma ^{\gcd (\frac{{({p^n} + 1)({p^k} + 1)}}{2},{p^{2n}} - 1)}}} \right\rangle$ and $C_1\cap C_2=\{\pm 1\}$. Let $\phi$ be a mapping defined on $C_1\cup C_2$ with $\phi(u)=\frac{1}{2}(u+u^{-1})$. It is easy to see that $\phi$ is a $2$-to-$1$ mapping on $(C_1\cup C_2)\setminus\{\pm 1\}$, and 
	\[\mathrm{Im}(\phi)\mid_{C_1\setminus\{\pm 1\}} \cap \mathrm{Im}(\phi)\mid_{C_2\setminus\{\pm 1\}}=\emptyset.\]
	We consider the following two cases. 
	
	{\it Case 1: }$ n $ is odd.	In this case we have $ \gcd(\frac{{({p^n} - 1)({p^k} + 1)}}{2},p^{2n}-1)=\frac{(p^{n}-1)(p+1)}{2} $ and $ \gcd(\frac{{({p^n} + 1)({p^k} + 1)}}{2},p^{2n}-1)=p^{n}+1 $. Then $ {C_1} = \left\langle {\gamma^{\frac{(p^{n}-1)(p+1)}{2}}} \right\rangle $, $ {C_2} = \left\langle {\gamma^{p^{n}+1}} \right\rangle $. Then the function $\theta^{p^k+1}$ is $(p+1)$-to-1 on $S_1$ and $2$-to-$1$ on $S_2$. 
	
	For any $b\in \mathrm{Im}(\phi)\mid_{C_1\setminus\{\pm 1\}}$, there are two $u$'s in $C_1\setminus\{\pm 1\}$ such that $\phi(u)=b$. Each $u$ corresponds $(p+1)$ $\theta$'s in $S_1$. Hence the $(-1)$-differential equation $\Delta_{-1}(x)=b$ has $\frac{p+1}{2}$ solutions. The number of such $b$ is $\frac{p^n+1}{p+1}-1$ since  $\#\{{C_1\setminus\{\pm 1\}}\}=2(\frac{p^n+1}{p+1}-1)$. For $b\in \mathrm{Im}(\phi)\mid_{C_1\setminus\{\pm 1\}}$, we can discuss it similarly. 
	
	For $b=1$ (respectively, $b=-1$), there is a unique $u=1$ (respectively, $u=-1$) such that $\phi(u)=b$. 
	Each $u$ corresponds $(p+1)$ $\theta$'s in $S_1$ and two $\theta$'s in $S_2$, i.e. $\theta\in \{\theta^{p+1}=1 :\theta\in S_1\}\cup \{\theta^{2}=1 :\theta\in S_2\}$ (respectively, $\theta\in \{\theta^{p+1}=-1 :\theta\in S_1\}\cup \{\theta^{2}=-1 :\theta\in S_2\}$).
	Note that {{four}} of them satisfy $\theta^2=1$ (respectively, $\theta^{2}=-1$), the other $(p-1)$ $\theta$'s do not satisfy $\theta^2=\pm 1$ since $ p \equiv 1\left( {\bmod~4} \right) $. Based on the corresponding relation of $x$ and $\theta$, we conclude that $\Delta_{-1}(x)=1$ has $\frac{p-1}{4}+1$ solutions. Summarizing the discussions as above, we have the following table.
		\begin{table}[!htbp]
		\begin{center}
			\begin{minipage}{250pt}
				\caption{}\label{table-3}
				\begin{tabular}{@{}llcl@{}}
					\toprule
					Set  $\mathbb{S}$&$ \# ~\mathbb{S} $&$ \# \{ {x \in {\gf_{{p^n}}}: {{\Delta _{-1}}(x) = b}.} \}, b\in\mathbb{S} $ \\
					\midrule
					$ \mathrm{Im}(\phi)\mid_{C_1\setminus\{\pm 1\}} $&$\frac{p^{n}-p}{p+1}$&$ \frac{p+1}{2} $\\
					$ \mathrm{Im}(\phi)\mid_{C_2\setminus\{\pm 1\}} $& $\frac{p^{n}-3}{2}$&$ 1 $\\
					$\{1\}$& $1$&$ \frac{p+3}{4} $\\
					$\{-1\}$& $ 1$&$ \frac{p+3}{4} $ \\
					$\gf_{{p^n}}\setminus$ the above & $\frac{{({p^n} + 1)(p-1)}}{2(p+1)}$& $0$\\
					\botrule
				\end{tabular}
			\end{minipage}
		\end{center}
	\end{table}

	{\it Case 2: }$ n $ is even.	
	In this case we have $ \gcd(\frac{(p^n-1)(p^k+1)}{2},p^{2n}-1)=p^{n}-1 $ and $ \gcd(\frac{(p^n+1)(p^k+1)}{2},p^{2n}-1)=\frac{(p^{n}+1)(p+1)}{2} $, then $ {C_1} = \left\langle {\gamma^{p^{n}-1}} \right\rangle $, $ {C_2} = \left\langle {\gamma^{\frac{(p^{n}+1)(p+1)}{2}}} \right\rangle $. Then the function $\theta^{p^k+1}$ is $2$-to-$1$ on $S_1$ and $(p+1)$-to-$1$ on $S_2$. By a similar discussion, we obtain the following table. We finish the proof.
	\begin{table}[!htbp]
		\begin{center}
			\begin{minipage}{250pt}
				\caption{}\label{table-4}
				\begin{tabular}{@{}llcl@{}}
					\toprule
					Set  $\mathbb{S}$&$ \# ~\mathbb{S} $&$ \# \{ {x \in {\gf_{{p^n}}}: {{\Delta _{-1}}(x) = b}.} \}, b\in\mathbb{S} $ \\
					\midrule
					$ \mathrm{Im}(\phi)\mid_{C_1\setminus\{\pm 1\}} $&$\frac{p^{n}-1}{2}$&$ 1 $\\
					$ \mathrm{Im}(\phi)\mid_{C_2\setminus\{\pm 1\}} $& $\frac{p^{n}-p-2}{p+1}$&$ \frac{p+1}{2} $\\
					$\{1\}$& $1$&$ \frac{p+3}{4} $\\
					$\{-1\}$& $ 1$&$ \frac{p+3}{4} $ \\
					$\gf_{{p^n}}\setminus$ the above & $\frac{{({p^n} - 1)(p-1)}}{2(p+1)}$& $0$\\
					\botrule
				\end{tabular}
			\end{minipage}
		\end{center}
	\end{table}
	\end{proof}

\begin{theorem}	\label{poddp+1/2-3}	
	Let $ F(x)=x^{\frac{p^{k}+1}{2}} $ be a power function over $ \gf_{p^{n}}$, where {$ p>7 $}, $ p \equiv 3\left( {\bmod~4} \right) $, $ \gcd(n,k)=1 $, and $ \frac{2n}{\gcd(2n,k)} $ is even. If $ n $ is odd, the $ (-1) $-differential spectrum of $F(x)$ is 
	\[\begin{gathered}
		\mathbb{S} = \left\{ {_{ - 1}{\omega _0} = \frac{{{p^n}(3p - 1) - (p + 5)}}{{4(p + 1)}}{,_{ - 1}}{\omega _2} = \frac{{{p^n} - 3}}{4}{,_{ - 1}}{\omega _{\frac{{p + 1}}{4}}} = 1,} \right. \hfill \\
		~~~~~~{\text{ }}\left. {_{ - 1}{\omega _{\frac{{p + 5}}{4}}} = 1{,_{ - 1}}{\omega _{\frac{{p + 1}}{2}}} = \frac{{{p^n} - p}}{{p + 1}}} \right\}. \hfill \\ 
	\end{gathered} \]
	If $ n $ is even, the $ (-1) $-differential spectrum of $F(x)$ is 
	\[\begin{gathered}
		S = \left\{ {_{ - 1}{\omega _0} = \frac{{({p^n} - 1)(3p - 1)}}{{4(p + 1)}}{,_{ - 1}}{\omega _2} = \frac{{{p^n} - 1}}{4}{,_{ - 1}}{\omega _{\frac{{p + 1}}{4}}} = 1,} \right. \hfill \\
		~~~~~~{\text{ }}\left. {_{ - 1}{\omega _{\frac{{p + 5}}{4}}} = 1{,_{ - 1}}{\omega _{\frac{{p + 1}}{2}}} = \frac{{{p^n} - p - 2}}{{p + 1}}{\text{ }}} \right\}. \hfill \\ 
	\end{gathered} \]
	
\end{theorem}
The proof is similar to that of Theorem \ref{poddp+1/2-1} and we omit it.

\section{A new class of AP$c$N power permutations and
	their $c$-differential spectra}

Very recently, the usual differential properties of the power permutation $F(x)=x^{\frac{5^n-3}{2}}$ over $\gf_{{5^n}}$ were studied in \cite{YHDLCM}. In this section, we prove that  $F(x)$ is AP$c$N when $c=-1$. The $(-1)$-differential spectrum of $F(x)$ is also given. First,  we investigate the $(-1)$-differential uniformity of $ x^{\frac{5^n-3}{2}} $. We introduce the following lemma. The proof is very similar to that of Lemma $ 3 $ in \cite{YHDLCM} and we omit it.

\begin{lemma}\label{p=5lemma}
	Let $ b \in {\gf_{{5^n}}}\backslash \left\{ { \pm 1} \right\}  $ be a nonzero nonsquare element. If both of the two quadratic equations $ x^2+x-b^{-1}=0 $ and $ y^2+y+b^{-1}=0 $ have solutions in $ \gf_{{5^n}} $, then the solution $ z \in \gf_{{5^n}} $ of the quadratic equation
		\[z^2 + (1-2b^{-1})z - b^{-1}=0\]
	satisfies $ \chi(z(z+1))=-1 $.
\end{lemma}
Based on the above lemma, we have the following theorem.
\begin{theorem} \label{p=5theorem}
	Denote by $ _{-1}\Delta_{F} $ the $ (-1) $-differential uniformity of $ F(x)=x^d $ over $ \gf_{{5^n}} $, where $ d=\frac{5^n-3}{2} $. We have $ _{-1}\Delta_{F}=2 $.
	
\end{theorem}

\begin{proof}
	For $ b \in \gf_{{5^n}} $, we consider the $(-1)$-differential equation \begin{equation} \label{p=5main}
		(x+1)^{d}+x^{d}=b
	\end{equation}
	over $ \gf_{{5^n}} $. It is obvious that when $b=0$, (\ref{p=5main}) has a unique solution $x=-\frac{1}{2}$ since $d$ is odd and  $ \gcd(d,5^n-1)=1 $.	In the following, we may assume that $b\neq 0$. Recall that $ \chi $ denotes the quadratic multiplicative character of $ \gf_{{5^n}} $. For all $x\in\gf_{{p^n}}^\# $, (\ref{p=5main}) becomes 
	\begin{equation}\label{p=5simplification}
		\chi(x+1)(x+1)^{-1}+\chi(x)x^{-1}=b.
	\end{equation}
	Depending on the values of $ \chi(x)$ and  $\chi(x+1) $,  we have four quadratic equations. By solving these four equations, we obtain Table \ref{table5}.
	\begin{table}[!htbp]\label{table5}
		\begin{center}
			\begin{minipage}{4250pt}
				\caption{Simplification of (\ref{p=5simplification}) in four cases}\label{table-5}
				\begin{tabular}{@{}llllll@{}}
					\toprule
					Case &  $ \chi \left( x \right) $ & $ \chi \left( x+1 \right) $ & Equation & $ x_1, x_2 $ & $ x_{1}x_{2} $  \\
					\midrule
					$ \rm\uppercase\expandafter{\romannumeral1} $ &  1 & 1 & $ x^2+(1-2b^{-1})x-b^{-1}=0 $ & $ \frac{-(1-2b^{-1})\pm \sqrt{1-b^{-2}}}{2} $ & $ -b^{-1} $  \\
					$ \rm\uppercase\expandafter{\romannumeral2} $ & 1 & -1 & $ x^2 + x - b^{-1} = 0 $ & $ \frac{-1\pm\sqrt{1-b^{-1}}}{2} $ & $ -b^{-1} $  \\
					$ \rm\uppercase\expandafter{\romannumeral3} $ &  -1 & 1 & $ x^2 + x + b^{-1} = 0 $ & $ \frac{-1\pm\sqrt{1+b^{-1}}}{2} $ & $ -b^{-1} $  \\
					$ \rm\uppercase\expandafter{\romannumeral4} $ &  $-1$ & $-1$ & $  x^2+(1+2b^{-1})x+b^{-1}=0 $ & $ \frac{-(1+2b^{-1})\pm \sqrt{1-b^{-2}}}{2} $ & $ b^{-1} $  \\
			
					\botrule
				\end{tabular}
			\end{minipage}
		\end{center}
	\end{table}

	Next we discuss the possible solutions of (\ref{p=5simplification}) for a given $b\neq \pm 1$. Note that $ \chi(-1)=1 $, if $ b $ is a square element, then $\chi(-b^{-1})=\chi(b^{-1})=1 $. In both Cases $ \rm\uppercase\expandafter{\romannumeral2} $ and $ \rm\uppercase\expandafter{\romannumeral3} $, if $ x $ is a solution, then $ \chi(x(x+1))=-1 $. So (\ref{p=5simplification}) has no solution in $S_{1, -1}$ and $S_{-1, 1}$. Now we consider Cases $ \rm\uppercase\expandafter{\romannumeral1} $ and $ \rm\uppercase\expandafter{\romannumeral4} $. Let $ x_{1,2}= \frac{-(1-2b^{-1}) \pm \sqrt{1-b^{-2}}}{2}  $ and $ x_{3,4}=\frac{-(1+2b^{-1}) \pm \sqrt{1-b^{-2}}}{2} $ be solutions of $ x^2+(1-2b^{-1})x-b^{-1}=0 $ and $ x^2+(1+2b^{-1})x+b^{-1}=0 $, respectively. Then $x_3=-(x_2+1)$ and $x_4=-(x_1+1)$. We assert that the number of solutions of (\ref{p=5simplification}) in Cases I and IV is at most two since 
	$\chi(x_1+1)=\chi(x_2+1)=1$ but $\chi(x_3)=\chi(x_4)=-1$. Then (\ref{p=5simplification}) has at most two solutions when $b$ is a square element. 
	
	If $ b $ is a nonsquare element, then $\chi(-b^{-1})=\chi(b^{-1})=-1 $. We also consider the solutions of (\ref{p=5simplification}) in each case. In Case $ \rm\uppercase\expandafter{\romannumeral1} $, the product of two solutions of equation $ x^2+(1-2b^{-1})x-b^{-1}=0 $ is $ -b^{-1} $, which is a nonsquare element, this means (\ref{p=5simplification}) has at most one solution in $ S_{1,1} $. Similarly, we can prove that (\ref{p=5simplification}) has at most one solution in $ S_{-1,-1} $. Now we consider Case $ \rm\uppercase\expandafter{\romannumeral2} $. Let $ x_{5} $ and $ -x_{5}-1 $ be the two solutions of the quadratic equation $ x^2+x-b^{-1}=0 $. It is easy to check that $ x_{5} \in S_{1,-1} $ if and only if $ -x_{5}-1 \in S_{1,-1} $. This implies that (\ref{p=5simplification}) has at most one solution in $ S_{1,-1} $. Similarly, we can prove that (\ref{p=5simplification}) has at most one solution in $ S_{-1,1} $. If (\ref{p=5simplification}) has solutions in Cases $ \rm\uppercase\expandafter{\romannumeral1} $ and $ \rm\uppercase\expandafter{\romannumeral3} $ simultaneously, then the discriminates of quadratic euqations $ x^2+(1-2b^{-1})x-b^{-1}=0 $ and $ x^2 + x + b^{-1} = 0 $ are both square elements, i.e., $ \chi(1-b^{-2})=\chi(1+b^{-1})=1 $. Then $ \chi(1-b^{-1})=1 $, which implies that the equation $ x^{2}+x+b^{-1}=0 $ has two solutions in $ \gf_{{5^n}} $. This contradicts Lemma \ref{p=5lemma}, since the solution $ x $ in Case $  \rm\uppercase\expandafter{\romannumeral1} $ satisfies $ \chi(x(x+1))=1 $. Then $ (\ref{p=5simplification}) $ cannot have solutions in Cases $ \rm\uppercase\expandafter{\romannumeral1} $ and $ \rm\uppercase\expandafter{\romannumeral3} $ simultaneously. Similarly, we can prove that (\ref{p=5simplification}) has no solution in Cases $ \rm\uppercase\expandafter{\romannumeral4} $ and $ \rm\uppercase\expandafter{\romannumeral3} $ simultaneously. Since Case $ \rm\uppercase\expandafter{\romannumeral2} $ has solutions if and only if Case $ \rm\uppercase\expandafter{\romannumeral3} $ has solutions, we conclude that for $ b \in {\gf_{{5^n}}}\backslash \left\{ { \pm 1} \right\}  $ (\ref{p=5main}) has solutions in at most two cases, i.e., $ \delta_{-1}(b) \le 2 $.

	For $ b=1 $, one can easy check that $ x=0 $ is a solution of (\ref{p=5main}). Then we consider (\ref{p=5simplification}) in the four cases. It is obvious that Cases $ \rm\uppercase\expandafter{\romannumeral2} $ and $ \rm\uppercase\expandafter{\romannumeral3} $ has no solution. In Case $ \rm\uppercase\expandafter{\romannumeral1} $, (\ref{p=5simplification}) becomes $ x^2-x-1=0 $. This quadratic equation has a unqiue solution $ x=-2 $. We know that $ \chi(-2)=-1 $ for odd $ n $ and $ \chi(-2) = 1 $ for even $ n $. Then Case $ \rm\uppercase\expandafter{\romannumeral1} $ has one solution when $ n $ is even and no solution when $ n $ is odd. In Case $ \rm\uppercase\expandafter{\romannumeral4} $, (\ref{p=5simplification}) becomes $ x^2-2x+1=0 $. This quadratic equation has a unique solution $ x=1 $. Therefore (\ref{p=5simplification}) has no solution in $ S_{-1,-1} $ when $ b=1 $. We conclude that $ \delta_{-1}(1)=2 $ for even $ n $ and $ \delta_{-1}(1)=1 $ for odd $ n $. 
	For $ b=-1 $, we can similarly obtain that $ \delta_{-1}(-1)=2 $ for even $ n $ and $ \delta_{-1}(-1)=1 $ for odd $ n $. The proof is finished.
\end{proof}


Next we will determine the $ (-1) $-differential spectrum of $ F(x)=x^{\frac{5^n-3}{2}} $. It is sufficient to determine  $ {}_{-1}N_{4} $, which denotes the number of solutions in $ (\gf_{{5^n}} )^4 $ of the equation system (\ref{equationsystem}) when $ c=-1 $. Moreover, we need the value of the character sum 
\[ {\Gamma _{5,n}}=\sum_{x\in\gf_{5^n}}\chi{(x(x-1)(x+1))},\]
where $\chi$ is the quadratic multiplicative character over $\gf_{{5^n}}$. It was proved in \cite{YHDLCM} that  
\begin{equation}\label{Gamma_{(5,n)}}
{\Gamma _{5,n}} = {\left( { - 1} \right)^{n + 1}}\sum\limits_{k = 0}^{\left\lfloor {\frac{n}{2}} \right\rfloor } {{{\left( { - 1} \right)}^k} \binom{n}{2k} {2^{2k + 1}}},
\end{equation}
which is always an integer. We have the following lemma on the value of $ {}_{-1}N_{4} $,
which is the number of the solutions $(x_1,x_2,x_3,x_4)\in(\gf_{{5^n}})^4$ of the equation system
\begin{equation}
\Bigg\{ \begin{array}{l}
	{x_1} - {x_2} + {x_3} - {x_4} = 0\\
	x_1^d + x_2^d - x_3^d - x_4^d = 0
\end{array}. 
\end{equation} 

\begin{lemma}\label{N_4value}
We have \[{{}_{-1}N_4} = \left\{ \begin{array}{l}
	1 + 4\left( {{5^n} - 1} \right) + ({5^n} - 1)\left( { - \frac{1}{4}{\Gamma _{5,n}} + \frac{{7 \cdot {5^n} - 17}}{4}} \right),~\mathrm{if}~n~\mathrm{is}~\mathrm{even},\\
	1 + 2\left( {{5^n} - 1} \right) + ({5^n} - 1)\left( { - \frac{1}{4}{\Gamma _{5,n}} + \frac{{7 \cdot {5^n} - 17}}{4}} \right),~\mathrm{if}~n~\mathrm{is}~\mathrm{odd},
\end{array} \right.\]
where $ \Gamma_{5,n} $ was given in (\ref{Gamma_{(5,n)}}).
\end{lemma}

\begin{proof}Since $ d $ is odd, the number of solutions $ (x_1,x_2,x_3,x_4) \in (\gf_{{5^n}} )^4 $ of the equation system 
\begin{equation}\label{equationsystemmain}
	\left\{ \begin{array}{l}
		{x_1} + {x_2} + {x_3} + {x_4} = 0\\
		{x^d_1} - {x^d_2} - {x^d_3} + {x^d_4} = 0
	\end{array} \right.
\end{equation}
is also $ {}_{-1}N_{4} $.
For a solution $ (x_1,x_2,x_3,x_4) $ of (\ref{equationsystemmain}), first we consider that there exists $ x_i=0 $ for some $ 0 \le i \le 4 $. It is easy to see that $ (0,0,0,0) $ is a solution of (\ref{equationsystemmain}), and (\ref{equationsystemmain}) has no solution containing only three zeros. If there are only two zeros in  $ (x_1,x_2,x_3,x_4) $, one can get that $(x,0,0,-x)$ and $(0,x,-x,0)$ are solutions of (\ref{equationsystemmain}), where $x\in\gf_{{5^n}}^*$.
That is, (\ref{equationsystemmain}) has $ 2(5^n-1) $ solutions containing only two zeros. We consider that there is only one zero in $ (x_1,x_2,x_3,x_4) $.  If $ x_{4}=0 $, then $ x_1,x_2 $ and $ x_3 $ are nonzero and they satisfy 
\[\left\{ \begin{array}{l}
	{x_1} + {x_2} + {x_3} = 0\\
	x^d_1 - x^d_2- x^d_3 = 0.
\end{array} \right.\]
Let $ y_i=\frac{x_i}{x_3} $ for $ i=1,2 $, we have $ y_1+y_2+1=0 $ and $ {y^d_1}-{y^d_2}-1=0 $ with $ y_1,y_2 \ne 0 $. Then $ y_2=-y_1-1 $ and $ (y_1+1)^d+y^d_1=1 $. By the proof of Theorem \ref{p=5theorem}, we know that equation $ (y_1+1)^d+y^d_1=1 $ has $ \delta_{-1}(1)-1 $ nonzero solutions in $ \gf_{{5^n}} $. Similarly, we can determine the number of solutions of (\ref{equationsystemmain}) with only $x_i=0$, $i=1,2$ and $3$. Then (\ref{equationsystemmain}) has $ 4(5^n-1)(\delta_{-1}(1)-1) $ solutions containing only one zero. We conclude that (\ref{equationsystemmain}) has $ 1+(5^n-1)(4\delta_{-1}(1)-2) $ solutions containing zeros. 

Next we consider $ x_i \ne 0 $ for $ 1 \le i \le 4 $. Let $ y_i = \frac{x_i}{x_4} $ for $ i=1,2,3 $. Then $ y_i \ne 0 $ and satisfy
\begin{equation}\label{systemsimple}
	\left\{ \begin{array}{l}
		{y_1} + {y_2} + {y_3} + 1 = 0\\
		{y^d_1} - {y^d_2} - {y^d_3} + 1 = 0.
	\end{array} \right.
\end{equation}
We denote by $ _{-1}n_4 $ the number of solutions of (\ref{systemsimple}). It can be seen that
\begin{equation}\label{N_4n_4relation}
	{}_{-1}N_4=1+(5^n-1)(4\delta_{-1}(1)-2)+(5^n-1){}_{-1}n_4.
\end{equation}
Recall that $ {y^d_i}=\chi(y_i)y^{-1}_i $, we discuss (\ref{systemsimple}) in the following cases.

{\it Case 1.} $ y_1, y_2 $ and $ y_3 $ are all square elements. Then (\ref{systemsimple}) becomes
\begin{equation}\label{n(1,1,1)}
	\left\{ \begin{array}{l}
		{y_1} + {y_2} + {y_3} + 1 = 0\\
		{y^{-1}_1} - {y^{-1}_2} - {y^{-1}_3} + 1 = 0.
	\end{array} \right.
\end{equation}
We denote by $ n_{(1,1,1)} $ the number of solution of (\ref{n(1,1,1)}). Next, we determine $ n_{(1,1,1)} $. If $ y_{1}=-1 $, then $ y_2=-y_3 $, $ (-1,y_{2},-y_{2}) $ is a solution of (\ref{n(1,1,1)}) when $ y_{2} $ is a square element. If $ y_{1}\ne-1 $, then $ y_{2}y_{3}=(y_{2}+y_{3})({{y^{-1}_2}+{y^{-1}_3}})^{-1}= -(y_{1}+1)({{y^{-1}_1}+1})^{-1}=-y_{1}$. Combining with the first equation of (\ref{n(1,1,1)}), we obtain $ y_{3}=\frac{y_{2}+1}{y_{2}-1} $ since $ y_{2}\ne1 $. Consequently, $ y_{1}=-\frac{y_{2}(y_{2}+1)}{y_{2}-1} $. If $ (y_{1},y_{2},y_{3}) $ is a desired solution, then $ \chi(\frac{y_{2}+1}{y_{2}-1})=\chi({y^2_2}-1)=1 $ and $ \chi(y_{2})=1 $. The number of such $ y_{2} $ is \[\begin{array}{l}
	\frac{1}{4}\sum\limits_{{y_2} \ne 0, \pm 1} {( {\chi ( {{y_2}} ) + 1} )( {\chi ( {{y^2_2} - 1} ) + 1} )} \\
	= \frac{1}{4}( {5^n} - 6+{\sum\limits_{{y_2} \in {\gf_{{5^n}}}} {\chi ( {y^{3}_2 - {y_2}} )}  + \sum\limits_{{y_2} \in {\gf_{{5^n}}}} {\chi ( {{y_2}} )}  + \sum\limits_{{y_2} \in {\gf_{{5^n}}}} {\chi ( {y_2^2 - 1} )  } })\\
	= \frac{1}{4}( {{\Gamma _{5,n}} + {5^n} - 7} ).
\end{array}\]
Note that there may exist $ y_{2} $ with $ \chi(y_{2})=\chi({y^2_2}-1)=1 $, such that $ y_{1}=-\frac{y_{2}(y_{2}+1)}{y_{2}-1}=-1 $, i.e., $ y_{2}=\pm 2 $. Two solutions $ (-1,2,-2) $ and $ (-1,-2,2) $ should be subtracted, this only occurs when $ n $ is even. We obtain
\[{n_{(1,1,1)}} = \left\{ \begin{array}{l}
	\frac{{3 \cdot {5^n} - 9}}{4} + \frac{1}{4}{\Gamma _{5,n}},~~\mathrm{if}~n~\mathrm{is}~\mathrm{odd}.\\
	\frac{{3 \cdot {5^n} - 17}}{4} + \frac{1}{4}{\Gamma _{5,n}},~\mathrm{if}~n~\mathrm{is}~\mathrm{even}.
\end{array} \right.\]

{\it Case 2.} Two of $ y_{1},y_{2} $ and $ y_{3} $ are square element, the other one is a nonsquare element. We first assume that $ \chi(y_{1})=\chi(y_{2})=1 $ and $ \chi(y_{3})=-1 $. Then (\ref{systemsimple}) becomes
\begin{equation}\label{n(1,1,-1)}
	\left\{ \begin{array}{l}
		{y_1} + {y_2} + {y_3} + 1 = 0\\
		{y^{-1}_1}- {y^{-1}_2} + {y^{-1}_3} + 1 = 0.
	\end{array} \right.
\end{equation}
We denote by $ n_{(1,1,-1)} $ the number of solution of (\ref{n(1,1,-1)}). Next, we determine $ n_{(1,1,-1)} $. Note that (\ref{n(1,1,-1)}) has no solution when $ y_{2}=\pm 1 $. If $ y_{2}\ne\pm1 $, we can obtain $ y_{1}+y_{3}\ne0 $, moreover,   $y_{1}y_{3}=(y_{1}+y_{3})({{y^{-1}_{1}}+{y^{-1}_{3}}})^{-1}= \left( {1 + {y_2}} \right)\frac{{{y_2}}}{{{y_2} - 1}}=\frac{y_{2}(y_{2}-1)}{y_{2}-1}$. It is mentioned that $ \chi(\frac{y_{2}+1}{y_{2}-1})=-1 $ since $ \chi(y_{1})=\chi(y_{2})=1 $ and $ \chi(y_{3})=-1 $. Moreover, $ y_{1} $ satisfies the following quadratic equation
\begin{equation}\label{n(1,1,-1)simple}
	{y^{2}_{1}}+(y_{2}+1)y_{1}+\frac{y_{2}(y_{2}+1)}{y_{2}-1}=0.
\end{equation}
The discriminate of (\ref{n(1,1,-1)simple}) is $ \Delta=(y_{2}+1)^2+\frac{y_{2}(y_{2}+1)}{y_{2}-1}=\frac{(y_{2}+1)(y_{2}-2)^2}{y_{2}-1}$. If (\ref{n(1,1,-1)simple}) has solutions in $ \gf_{{5^n}} $, then $ y_{2} $ must be $ 2 $ since $ \chi(\frac{y_{2}+1}{y_{2}-1})=-1 $. Consequently, $ y_{1}=y_{3}=1 $. Note that $ \chi(y_{3})=-1 $, we conclude that $ n_{(1,1,-1)}=0 $. Next we assume that $ \chi(y_{1})=\chi(y_{3})=1 $ and $ \chi(y_2)=-1 $. Then (\ref{systemsimple}) becomes \begin{equation}\label{n(1,-1,1)}
	\left\{ \begin{array}{l}
		{y_1} + {y_2} + {y_3} + 1 = 0\\
		{y^{-1}_1} + {y^{-1}_2} - {y^{-1}_3} + 1 = 0.
	\end{array} \right.
\end{equation}
We denote by $ n_{(1,-1,1)} $ the number of solution of (\ref{n(1,-1,1)}). Similarly, we can get that $ n_{(1,-1,1)}=0 $. Next, we assume that $ \chi(y_{2})=\chi(y_{3})=1 $ and $ \chi(y_1)=-1 $. Then (\ref{systemsimple}) becomes \begin{equation}\label{n(-1,1,1)}
	\left\{ \begin{array}{l}
		{y_1} + {y_2} + {y_3} + 1 = 0\\
		-{y^{-1}_1} - {y^{-1}_2} - {y^{-1}_3} + 1 = 0.
	\end{array} \right.
\end{equation}
We denote by $ n_{(-1,1,1)} $ the number of solution of (\ref{n(-1,1,1)}). Let $ z_{1}=\frac{y_{2}}{y_{3}} $, $ z_{2}=\frac{y_{1}}{y_{3}} $ and $ z_{3}=\frac{1}{y_{3}} $. Then $ \chi(z_{1})=\chi(z_{3})=1 $ and $ \chi(z_{2})=-1 $. From (\ref{n(-1,1,1)}) we obtain \begin{equation}\label{n(-1,1,1)trans}
	\left\{ \begin{array}{l}
		{z_1} + {z_2} + {z_3} + 1 = 0\\
		{z^{-1}_1} + {z^{-1}_2} - {z^{-1}_3} + 1 = 0.
	\end{array} \right.
\end{equation}
The number of solutions of (\ref{n(-1,1,1)trans}) is $ n_{(1,-1,1)} $, i.e., $ n_{(-1,1,1)}=n_{(1,-1,1)}=0 $.

{\it Case 3.} Two of $ y_{1},y_{2} $ and $ y_{3} $ are nonsquare element, the other one is a square element. We first assume that $ \chi(y_{1})=1 $ and $ \chi(y_{2})=\chi(y_{3})=-1 $. Then (\ref{systemsimple}) becomes
\begin{equation}\label{n(1,-1,-1)}
	\left\{ \begin{array}{l}
		{y_1} + {y_2} + {y_3} + 1 = 0\\
		{y^{-1}_1} + {y^{-1}_2} + {y^{-1}_3} + 1 = 0.
	\end{array} \right.
\end{equation}

We denote by $ n_{(1,-1,-1)} $ the number of solution of (\ref{n(1,-1,-1)}). Next, we determine $ n_{(1,-1,-1)} $. Note that $ y_{3}\ne-1 $, since $ \chi(y_{3})=-1 $. Then $ y_{1}+y_{2} \ne 0 $, moreover, $ y_{1}y_{2}=(y_{1}+y_{2})({y^{-1}_{1}}+{y^{-1}_{2}})^{-1}=(-y_{3}-1)({-y^{-1}_3}-1)^{-1}=y_{3} $. Then we obtain $ (y_{1}+1)(y_{2}+1)=0 $ from the first equation of (\ref{n(1,-1,-1)}). If (\ref{n(1,-1,-1)}) has solutions in $ \gf_{{5^n}} $, then $ y_{1} $ must be $ -1 $ since $ \chi(y_{2})=-1 $. Equation system (\ref{n(1,-1,-1)}) has solutions with types $ (-1,y_{2},-y_{2}) $, where $ \chi(y_{2})=-1 $. Then we have $ n_{(1,-1,-1)}=\frac{5^n-1}{2} $.

Then we assume that $ \chi(y_{2})=1 $ and $ \chi(y_{1})=\chi(y_{3})=-1 $. Then (\ref{systemsimple}) becomes
\begin{equation}\label{n(-1,1,-1)}
	\left\{ \begin{array}{l}
		{y_1} + {y_2} + {y_3} + 1 = 0\\
		-{y^{-1}_1} - {y^{-1}_2} + {y^{-1}_3} + 1 = 0.
	\end{array} \right.
\end{equation}

We denote by $ n_{(-1,1,-1)} $ the number of solution of (\ref{n(-1,1,-1)}). Next, we determine $ n_{(-1,1,-1)} $. It is similar to Case $ 1 $. We obtain $ y_{1}y_{2}=-y_{3} $ and $ y_{2}=\frac{y_{1}+1}{y_{1}-1} $ since $ y_{1}\ne \pm1 $. Consequently, $ y_{3}=-\frac{y_{1}(y_{1}+1)}{y_{1}-1} $. If $ (y_{1},y_{2},y_{3}) $ is a desired solution, then $ \chi(\frac{y_{1}+1}{y_{1}-1})=\chi({y^{2}_{1}}-1)=1 $ and $ \chi(y_{1})=-1 $. The number of such $ y_{1} $ is 

\[\begin{array}{l}
	- \frac{1}{4}\sum\limits_{{y_1} \ne 0, \pm 1} {( {\chi ( {{y_1}} ) - 1} )( {\chi ( {{y^2_1} - 1} ) + 1} )} \\
	=  - \frac{1}{4}( {\sum\limits_{{y_1} \in {\gf_{{5^n}}}} {\chi ( {y_1^3 - {y_1}} )}  + \sum\limits_{{y_1} \in {\gf_{{5^n}}}} {\chi ( {{y_1}} )}  - \sum\limits_{{y_1} \in {\gf_{{5^n}}}} {\chi ( {y_1^2 - 1} ) - {5^n} + 2} } )\\
	=  - \frac{1}{4}( {{\Gamma _{5,n}} - {5^n} + 3} ).
\end{array}\]

We obtain $ n_{(-1,1,-1)}= -\frac{1}{4}{\Gamma _{5,n}} + \frac{{{5^n} - 3}}{4} $. Then we assume that $ \chi(y_{3})=1 $ and $ \chi(y_{1})=\chi(y_{2})=-1 $. Then (\ref{systemsimple}) becomes
\begin{equation}\label{n(-1,-1,1)}
	\left\{ \begin{array}{l}
		{y_1} + {y_2} + {y_3} + 1 = 0\\
		{y^{-1}_1} - {y^{-1}_2} + {y^{-1}_3}-1 = 0.
	\end{array} \right.
\end{equation}

We denote by $ n_{(-1,-1,1)} $ the number of solution of (\ref{n(-1,-1,1)}). Similarly, we can get that  $ n_{(-1,-1,1)}= -\frac{1}{4}{\Gamma _{5,n}} + \frac{{{5^n} - 3}}{4} $. 

{\it Case 4.}  $ y_{1},y_{2} $ and $ y_{3} $ are all nonsquare elements, i.e., $ \chi(y_{1})=\chi(y_{2})=\chi(y_{3})=-1 $. Then (\ref{systemsimple}) becomes
\begin{equation}\label{n(-1,-1,-1)}
	\left\{ \begin{array}{l}
		{y_1} + {y_2} + {y_3} + 1 = 0\\
		-{y^{-1}_1} + {y^{-1}_2} + {y^{-1}_3} + 1 = 0.
	\end{array} \right.
\end{equation}

We denote by $ n_{(-1,-1,-1)} $ the number of solution of (\ref{n(-1,-1,-1)}). Next, we determine $ n_{(-1,-1,-1)} $. Note that $ y_{1}\ne-1 $, since $ \chi(y_{1})=-1 $. It is similar to Case $ 2 $. Since $ \chi(y_{1})=\chi(y_{2})=\chi(y_{3})=-1 $, we can get that (\ref{n(-1,-1,-1)}) has no solution in $ \gf_{{5^n}} $.

By discussions as above, ${{}_{-1}n_4} = \sum \limits_{{i,j,k} \in \{\pm1\}}{{n_{\left( {i,j,k} \right)}}} $ and then $ {}_{-1}N_{_4} $ follows by (\ref{N_4n_4relation}) and the value of $ \delta_{-1}(1) $.
\end{proof}

Then we can determine the $ (-1) $-differential spectrum of $ F(x)=x^{\frac{5^{n}-3}{2}} $.

\begin{theorem}\label{p=5spectrum}
The power function $ F(x)=x^{\frac{5^n-3}{2}} $ over $ \gf_{{5^n}} $ when $ c=-1 $ is  AP$c$N  with $ (-1)$-differential spectrum 
\[\begin{gathered}
	\mathbb{S} = \left\{ {_{ - 1}{\omega _0} =  - \frac{1}{8}{\Gamma _{5,n}} + \frac{{3 \cdot {5^n} - 5}}{8}{,_{ - 1}}{\omega _1} = \frac{1}{4}{\Gamma _{5,n}} + \frac{{{5^n} + 5}}{4},} \right. \hfill \\
	~~~~~~{\text{ }}\left. {_{ - 1}{\omega _2} =  - \frac{1}{8}{\Gamma _{5,n}} + \frac{{3 \cdot {5^n} - 5}}{8}{\text{ }}} \right\}. \hfill \\ 
\end{gathered} \]
when $ n $ is even, and with $ (-1) $-differential spectrum
\[\begin{gathered}
	\mathbb{S} = \left\{ {_{ - 1}{\omega _0} =  - \frac{1}{8}{\Gamma _{5,n}} + \frac{{3 \cdot {5^n} - 13}}{8}{,_{ - 1}}{\omega _1} = \frac{1}{4}{\Gamma _{5,n}} + \frac{{{5^n} + 13}}{4}} \right., \hfill \\
	~~~~~~{\text{ }}\left. {_{ - 1}{\omega _2} =  - \frac{1}{8}{\Gamma _{5,n}} + \frac{{3 \cdot {5^n} - 13}}{8}{\text{ }}} \right\}. \hfill \\ 
\end{gathered} \]
when $ n $ is odd, where $ \Gamma_{5,n} $ is determined in (\ref{Gamma_{(5,n)}}).
\end{theorem}
\begin{proof}
By (\ref{i^2omega_ieqution}) and Lemma \ref{N_4value}, we obtian that the elements in the $ (-1) $-differential spectrum satisfy
\[\sum\limits_{i = 0}^{{}_c{\Delta _F}} {{i^2}\cdot{{}_{-1}\omega _i}}  = \left\{ \begin{array}{l}
	- \frac{1}{4}{\Gamma _{5,n}} + \frac{{7 \cdot {5^n} - 5}}{4},~~\mathrm{if}~n~\mathrm{is}~\mathrm{even}.\\
	- \frac{1}{4}{\Gamma _{5,n}} + \frac{{7 \cdot {5^n} - 13}}{4},~\mathrm{if}~n~\mathrm{is}~\mathrm{odd}.
\end{array} \right.\]
Then by solving the equation (\ref{omegaiomega}) in Lemma \ref{properties1}. The proof is finished.
\end{proof}	

\section{concluding remarks}
In this paper, we mainly studied the $c$-differential spectra of power functions over finite fields. Some basic properties of the $c$-differential spectrum of a power function were given. The $c$-differential spectra of some classes of power functions were determined. Moreover, we proposed a class of AP$c$N function over $\gf_{{5^n}}$ with $c=-1$. Our future work is to find more power functions with low $c$-differential uniformity and to determine their $c$-differential spectra. It is worth finding applications of  these functions in sequences, coding theory and combinatorial designs.


\begin{thebibliography}{9}

\bibitem{bracken2010highly} C. Bracken, G. Leander, A highly nonlinear differentially $ 4 $ uniform power mapping that permutes fields of even degree, Finite Fields Appl., 16(4) (2010) 231-242.

\bibitem{BCC} C. Blondeau, A. Canteaut, P. Charpin, Differential properties of power functions, Int. J. Inf. Coding Theory, 1(2) (2010) 149-170.

\bibitem{BCC2} C. Blondeau, A. Canteaut, P. Charpin, Differential properties of $x \mapsto x^{2^t-1}$, IEEE Trans. Inf. Theory, 57(12) (2011) 8127-8137.

\bibitem{BP} C. Blondeau, L. Perrin, More differentially $6$-uniform power functions, Des. Codes Cryptogr., 73(2) (2014) 487-505.

\bibitem{BC} D. Bartoli, M. Calderini, On construction and (non) existence of $ c $-(almost) perfect nonlinear functions, Finite Fields Appl., 72 (2021) 101835.


\bibitem{2020MD} D. Bartoli, M. Timpanella, On a generalization of planar functions, Journal of Algebraic Combinatorics, 52 (2020) 187-213.
	
\bibitem{1991Differential} E. Biham, A. Shamir, Differential cryptanalysis of DES-like cryptosystems, J. Cryptol., 4(1) (1991) 3-72.

\bibitem{1993DifferentialBS} E. Biham, A. Shamir, Differential cryptanalysis of the data encryption standard, 
Springer, (1993).

\bibitem{2002BN} N. Borisov, M. Chew, R. Johnson, D. Wagner, Multiplicative differentials, International Workshop on Fast Software Encryption, 2365 (2002) 17-33.


\bibitem{CM1997} R.S. Coulter, R. Matthews, Planar functions and planes of Lenz-Barlotti class II, Des. Codes  Cryptogr., 10(2) (1997) 167-184.	

\bibitem{CHNC} S.-T. Choi, S. Hong, J.-S. No, H. Chung, Differential spectrum of some power functions in odd prime characteristic, Finite Fields Appl., 21 (2013) 11-29.


\bibitem{charpin2014sparse}
P. Charpin, G. M. Kyureghyan, V. Suder, Sparse permutations with low differential uniformity, Finite Fields Appl., 28 (2014) 214-243.


\bibitem{DWelch} H. Dobbertin, Almost perfect nonlinear power functions on $ GF(2^{n}) $: The Welch case, IEEE Trans. Inf. Theory, 45(4) (1999) 1271-1275.

\bibitem{DNiho} H. Dobbertin, Almost perfect nonlinear power functions on $ GF(2^{n}) $: The Niho case, IEEE Trans. Inf. Theory, 151(1-2) (1999) 57-72.

\bibitem{Dobbertin2001} H. Dobbertin, T. Helleseth, P. V. Kumar, H. Martinsen, Ternary $m$-sequences with three-valued cross-correlation function:  New decimations of Welch and Niho type, IEEE Trans. Inf. Theory, 47(4) (2001) 1473-1481.


\bibitem{DO1968} P. Dembowski, T. G. Ostrom, Planes of order $ n $ with collineation groups of order $ n^2 $, Math. Z, 103(3) (1968) 239-258.

\bibitem{2020CEP} P. Ellingsen, P. Felke, C. Riera, P. St{\u{a}}nic{\u{a}}, A. Tkachenko, C-differentials, multiplicative uniformity, and (almost) perfect c-nonlinearity, IEEE Trans. Inf. Theory, 66(9) (2020) 5781-5789.


\bibitem{hasan2021c} S.U. Hasan, M. Pal, C. Riera, P. St{\u{a}}nic{\u{a}}, On the $ c $-differential uniformity of certain maps over finite fields, Des. Codes  Cryptogr., 89(2) (2021) 221-239.

\bibitem{HRS} T. Helleseth, C. Rong, D, Sandberg, New families of almost perfect nonlinear power mappings, IEEE Trans. Inf. Theory, 45(2) (1999) 475-485.



\bibitem{HS} T. Helleseth, D. Sandberg, Some power mappings with low differential uniformity, Applicable Algebra in Eng. Commun. Comput., 8 (1997) 363–370.

 \bibitem{Li} N. Li, Y. N. Wu, X. Y. Zeng, X. H. Tang, On the differential spectrum of a class of power functions over finite fields, arXiv:2012.04316.
  
 \bibitem{LRF} L. Lei, W. Ren, C. L. Fan, The differential spectrum of a class of power functions over finite fields, Adv. Math. Commun., doi: 10.3934/amc.2020080.  
   
\bibitem{FF} R. Lidl, H. Niederreite, Finite Fields, Encyclopedia of Mathematics and Its Applications. vol. 20. Cambridge university press, Cambridge, (1997).


\bibitem{2021SYZ} S. Mesnager, C. Riera, P. St{\u{a}}nic{\u{a}}, H. D. Yan, Z. C. Zhou, Investigations on $ c $-(almost) perfect nonlinear functions, IEEE Trans. Inf. Theory, 67(10) (2021) 6916-6925.


\bibitem{NK1993} K. Nyberg, Differentially uniform mappings for cryptography, In: T. Helleseth (ed.) Advances in cryptology---EUROCRYPT'93. Norway. 1993. LNCS, vol. 765, pp. 55-64. Springer, Berlin, (1994).



\bibitem{tu2021class}Z. R. Tu, X. Y. Zeng, Y. P. Jiang, X. H. Tang, A class of AP$c$N power functions over finite fields of even characteristic, arXiv:2107.06464.

\bibitem{wang2021several} X. Q. Wang, D. B. Zheng, Several classes of P$c$N power functions over finite fields, arXiv:2104.12942.

\bibitem{WLZ} Y. N. Wu, N. Li, X. Y. Zeng, New P$c$N and AP$c$N functions over finite fields,  arXiv:2010.05396.

\bibitem{XYY} M. S. Xiong, H. D. Yan, P. Z. Yuan, On a conjecture of differentially $8$-uniform power functions, Des. Codes Cryptogr., 86(8) (2018) 1601-1621.
 
\bibitem{XY2017} M. S. Xiong, H. D. Yan, A note on the differential spectrum of a differentially $ 4 $-uniform power function, Finite Fields Appl., 48 (2017) 117-125.

\bibitem{2020XZL} Y. X. Xia, X. L. Zhang, C. L. Li, T. Helleseth, The differential spectrum of a ternary power mapping, Finite Fields Appl., 64 (2020) 101660.

\bibitem{yan2021on1} H. D. Yan, On $-1$-differential uniformity of ternary APN power functions, arXiv:2101.10543.

\bibitem{YHDLCM} H. D. Yan, C. J. Li, Differential spectra of a class of power permutations with characteristic $ 5 $, Des. Codes Cryptogr., 89 (2021) 1181-1191. 
 
\bibitem{YXLHXL} H. D. Yan, Y. B. Xia, C. L. Li, T. Helleseth, J. Q. Luo, The Differential Spectrum of the Power Mapping $ x^{p^ n-3} $, arXiv:2108.03088. 
 
 
\bibitem{Yan} H. D. Yan, Z. C. Zhou, J. Wen, J. M. Weng, T. Helleseth, Q. Wang, Differential spectrum of Kasami power permutations over odd characteristic finite fields, IEEE Trans. Inf. Theory, 65(10) (2019) 6819-6826.


\bibitem{2021SomeZha} Z. B. Zha, L. Hu, Some classes of power functions with low $ c $-differential uniformity over finite fields, Des. Codes Cryptogr., 89(1) (2021) 1193-1210.

\bibitem{ZKW2009} Z. B. Zha, G. M. Kyureghyan, X. Q. Wang, Perfect nonlinear binomials and their semifields, Finite Fields Appl., 15(2) (2009) 125-133.

\bibitem{ZW2009} Z. B. Zha, X. Q. Wang, New families of perfect nonlinear polynomial functions, J. Algebra, 332 (2009) 3912-3918. 

  


\end{thebibliography}
\end{document}